\documentclass[10pt,twocolumn,twoside]{IEEEtran}
\usepackage{amsfonts,dsfont,amssymb,amsbsy,amsmath,paralist,bm,ifthen,color}
\usepackage[pdfstartview=FitH,bookmarksnumbered,unicode,bookmarksopen=true]{hyperref}
\usepackage{graphicx}
\usepackage{epstopdf}
\usepackage{amsfonts} 
\usepackage{multirow}
\usepackage{subfigure,relsize}
\usepackage{mathrsfs}
\usepackage{bm}
\usepackage{bbm}
\usepackage{stfloats}
\usepackage{amsthm}
\usepackage{url}
\usepackage[noadjust]{cite}
\usepackage{cases,booktabs}
\usepackage{graphicx,algorithmic,algorithm,relsize}
\usepackage[justification=centering]{caption} 
\usepackage[table]{xcolor}

\newtheorem{lem}{Lemma}
\newtheorem{prop}{Proposition}

\newtheorem{rem}{Remark}

\newtheorem{definition}{Definition}

\newcommand\Cbb{\ensuremath{{\mathbb{C}}}}

\graphicspath{{fig_ga/}}

\definecolor{green}{RGB}{34	195	46}
\definecolor{red}{RGB}{220 0 0}

\usepackage{graphicx} 
\allowdisplaybreaks[4]

\title{Environment-Aware Network-Level Design of Pinching-Antenna Systems--Part II: \\Geometry-Aware Case}

\author{
 Yanqing Xu, \IEEEmembership{Senior Member, IEEE,}
         Zhiguo Ding, \IEEEmembership{Fellow, IEEE,}
         Xiu Yin Zhang, \IEEEmembership{Fellow, IEEE,}\\
         Trung Q. Duong, \IEEEmembership{Fellow, IEEE,}
         and Tsung-Hui Chang, \IEEEmembership{Fellow, IEEE}
         \thanks{\smaller[1] Part of this work has been accepted by IEEE GLOBECOM 2026 \cite{xu2026geometry}.}
         \thanks{\smaller[1] Y. Xu is with the School of Science and Engineering, The Chinese University of Hong Kong, Shenzhen, 518172, China (email: xuyanqing@cuhk.edu.cn).}
         \thanks{\smaller[1] Zhiguo Ding is with the School of Electrical \& Electronic Engineering, Nanyang Technological University, 639798, Singapore  (e-mail: zhiguo.ding@ntu.edu.sg).}
         \thanks{\smaller[1] X. Y. Zhang is with the School of Microelectronics, South China University of Technology, Guangzhou 510000, China, and also with the Pazhou Laboratory, Guangzhou 510330, China (e-mail: eexyz@scut.edu.cn).}
         \thanks{\smaller[1] T. Q. Duong is with the Faculty of Engineering and Applied Science, Memorial University, St. John's, NL A1C 5S7, Canada and also with the School of Electronics, Electrical Engineering and Computer Science, Queen's University Belfast, Belfast, U.K. (e-mail: tduong@mun.ca.)}
         \thanks{\smaller[1] T.-H. Chang is with the School of Artificial Intelligence, The Chinese University of Hong Kong, Shenzhen, 518172, China (email: changtsunghui@cuhk.edu.cn).} 
        }

\date{\today}

\begin{document}

\maketitle
\begin{abstract}
    This two-part paper develops an environment-aware network-level design framework for pinching-antenna systems, complementing existing link-level designs that optimize transmission according to instantaneous user configurations. While link-level optimization provides short-term user-centric adaptation, the proposed network-level framework configures the spatial radiation infrastructure according to longer-term environment information and evaluates performance over the service region. Part I investigates the traffic-aware case, where user presence is characterized statistically by a spatial traffic map and deployments are optimized using traffic-aware network-level metrics. Part II complements Part I by developing geometry-aware, blockage-aware network optimization for pinching-antenna systems in obstacle-rich environments. We introduce a grid-level average signal-to-noise (SNR) model with a deterministic LoS visibility indicator and a discrete activation architecture, where the geometry-dependent terms are computed offline in advance. Building on this model, we formulate two network-level activation problems: (i) average-SNR-threshold coverage maximization and (ii) fairness-oriented worst-grid average-SNR maximization. On the algorithmic side, we prove the coverage problem is NP-hard  and derive an equivalent mixed-integer linear programming reformulation through binary coverage variables and linear SNR linking constraints. To achieve scalability, we further develop a structure-exploiting coordinate-ascent method that updates one waveguide at a time using precomputed per-candidate SNR contributions. For the worst-grid objective, we adopt an epigraph reformulation and leverage the resulting monotone feasibility in the target SNR, motivating a low-complexity bisection-based method with coordinate-wise feasibility search over the discrete candidate set. Simulation results validate the proposed designs and quantify their gains under different environments and system parameters.
\end{abstract}

\begin{IEEEkeywords}
     Pinching antenna, environment-aware design, average-SNR-threshold coverage maximization, worst-grid average-SNR maximization.
\end{IEEEkeywords}

\section{Introduction} 
In addition to link-level communication performance, modern wireless deployments are increasingly evaluated by how well the network serves the whole region. Operators are expected to provide broadly uniform service, avoiding persistent weak-coverage areas and large location-to-location disparities, under practical constraints such as limited site availability, irregular layouts, and evolving usage patterns. This motivates network-level design approaches that treat performance as a spatial quantity and evaluate service quality over a region rather than at a small set of scheduled user locations \cite{niu2011tango,bi2019engineering,xu2025distributed}.
Despite this shift, current network-level planning is still largely built upon fixed-site infrastructures, where base stations (BS) and antenna panels are installed at predetermined locations and the network can mainly adapt through antenna tilt/beam patterns, transmit powers, and scheduling \cite{partov2014utility,buenestado2016self,luo2023srcon,zeng2024tutorial,li2026pemnet,peng2025rf}. While effective in many settings, fixed-site deployments inherently lack spatial agility. Especially, when the environment is irregular (e.g., cluttered indoor spaces, partitioned halls, industrial plants, or dense urban canyons), location-dependent propagation conditions can create persistent weak-coverage regions that are difficult to remedy without costly site densification or infrastructure reconfiguration. In such environments, link-level and network-level design address different but complementary objectives: the former adapts transmission to the currently served users, whereas the latter shapes the spatial service capability of the infrastructure over the entire region.

Recently, flexible-antenna techniques have been proposed to provide
additional spatial degrees of freedom for wireless transmission
\cite{wong2020fluid,zhu2025tutorial,shao20246d,shao20256d}.
In this context, pinching antennas introduce a different network-level degree of freedom by enabling reconfigurable radiation points along spatially long dielectric waveguides \cite{ding2024flexible}. By adjusting where energy is effectively radiated or collected along these structures, the network can shift the transmit/receive locations within the service region in a flexible and lightweight manner. This capability complements classical beam-domain adaptation by providing a new lever in the geometry domain. Instead of only shaping beams from fixed sites, the network can also reshape the physical emission/reception locations to better align with region-wide service objectives. Consequently, pinching antennas open up new opportunities for environment-aware network optimization, enabling deployments that improve region-wide coverage and robustness without relying solely on dense fixed-site installations.

This two-part work aims to develop an environment-aware, network-level design methodology for  pinching-antenna systems by capturing two fundamental drivers of spatial performance variability. Specifically, Part I focuses on the traffic-aware setting, where user presence is modeled statistically through a spatial traffic map and network performance is assessed via traffic-aware region-wide metrics. However, traffic heterogeneity is only one side of the story. Even with identical demand distributions, service quality can differ sharply across locations due to the propagation environment, which determines where reliable links are feasible in the first place.

In many practical scenarios, such as partitioned indoor spaces, warehouses, shopping malls, and industrial floors-walls, partitions, shelves, and other obstacles create strongly geometry-dependent channel conditions. These structures give rise to line-of-sight (LoS) corridors, abrupt shadow boundaries, and persistent coverage holes that cannot be captured by traffic statistics alone. An example is shown in Fig. \ref{fig: system model siso}(a), which shows that the conventional fixed BS cannot guarantee a satisfactory region-wide coverage.
This motivates Part II, which studies geometry-aware network-level pinching-antenna design. Here, ``geometry-aware'' means that the deployment explicitly incorporates blockage-induced LoS/non-line-of-sight (NLoS) transitions (or LoS feasibility) over the service region based on environment information (e.g., obstacle maps), enabling optimization of region-wide service quality without relying on instantaneous user geometry.

\begin{figure*}[!t]
	\centering
	\includegraphics[width=0.99\linewidth]{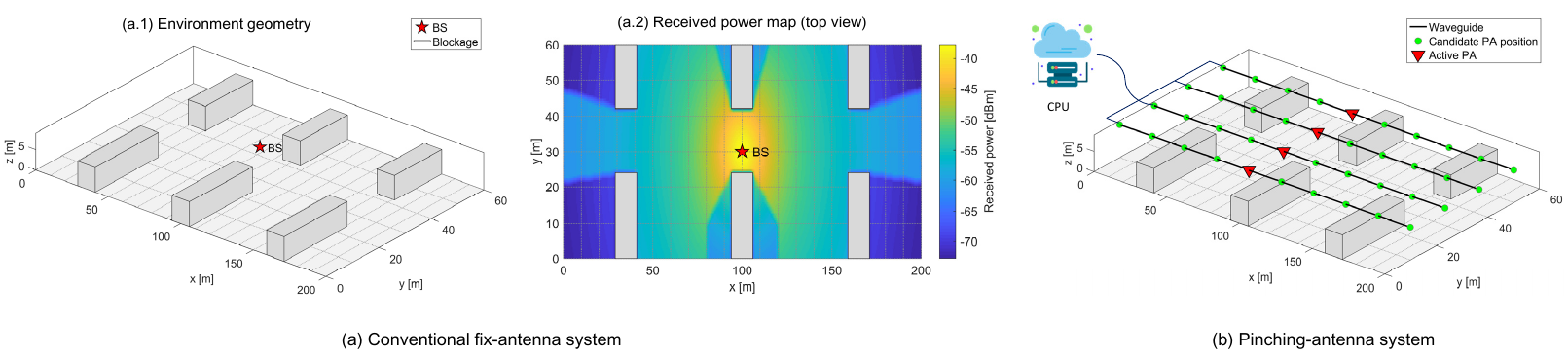}\\
        \captionsetup{justification=justified, singlelinecheck=false, font=small}	
        \caption{Illustration of the considered indoor environment and antenna architectures: (a.1) three-dimensional environment geometry with blockages and a fixed BS at the center; (a.2) corresponding top-view received power map for the conventional fixed-antenna system; (b) pinching-antenna system with ceiling waveguides, candidate pinching antenna (PA) positions, and selected active PAs coordinated by a central processing unit (CPU).} \label{fig: system model siso} \vspace{-0mm}
\end{figure*}

Beyond the numerous studies that focus on LoS-dominated scenarios and demonstrate the benefits of pinching antennas in improving spectral efficiency \cite{xu2025rate,tegos2025minimum}, enhancing energy efficiency \cite{zeng2025energy,zhu2026se}, and reducing outage probability \cite{tyrovolas2025performance,cheng2025performance}, a growing body of work has begun to investigate blockage-aware modeling and optimization for pinching-antenna systems \cite{ding2025blockage,wang2025blockage,li2025power,zhao2026robust,sun2026stochastic,jiang2025spatially,hu2025average,xu2025pinching-nlos,xu2025losblockage,xie2026pinching,ding2025EDMA}.
Specifically, a first line of studies adopted the probabilistic channel model to characterize LoS blockages and NLoS fading and revealed how blockage reshapes performance trends and design insights. Representative examples include analyzing the impact of LoS blockage on link reliability/outage behavior and developing optimization formulations that explicitly incorporate probabilistic LoS availability into pinching-antenna placement and beamforming for given user sets \cite{ding2025blockage,li2025power,zhao2026robust}.
Building on these models, a second line of works developed blockage-aware pinching-antenna system designs under diverse service objectives. For instance, blockage-aware pinching-antenna designs were investigated for simultaneous wireless information and power transfer (SWIPT) by jointly optimizing the pinching-antenna position and the power-splitting ratio to maximize the average SNR subject to energy-harvesting constraints \cite{jiang2025spatially}. In parallel, the interplay between blockage-aware modeling and hardware-level effects (e.g., in-waveguide attenuation) was studied in \cite{xu2025losblockage}, which showed that neglecting in-waveguide attenuation may incur only an insignificant performance loss under typical probabilistic blockage conditions, and proposed a dynamic sample-average-approximation algorithm for stochastic nonconvex multi-user designs.
Beyond purely statistical blockage models, more recent efforts consider geometry-/map-aware blockage to better reflect obstacle-rich indoor venues. A representative work \cite{xie2026pinching} developed a deterministic cylinder-obstacle blockage model to construct explicit LoS/blocked regions and, based on this geometry-aware model, both discrete and continuous pinching-antenna placement were studied, showing that obstacle-aware spatial reconfigurability can improve performance over fixed-antenna baselines. By leveraging the ability of pinching antennas to reconfigure LoS links, an environment-division multiple access (EDMA) technique was proposed in \cite{ding2025EDMA}, illustrating the potential of pinching antennas to support multi-user communications through environment-aware spatial adaptation. Considering practical implementation constraints, recent works have
investigated discrete/two-state pinching-antenna architectures with
pre-installed radiation points, including activation-state selection
\cite{galanopoulou2026viterbi} and the performance impact of realizing continuously
adjustable pinching-antenna positions using a finite number of discrete
candidate antennas \cite{tyrovolas2026many}.
More recently, generalized pinching-antenna systems have been
investigated beyond the conventional dielectric-waveguide implementation \cite{xu2025generalized}.
For example, a leaky-coaxial-cable-based GPA architecture was studied in
\cite{wang2025generalized}, where controllable radiation slots are exploited for flexible
wireless access, while a radio-stripe-based realization was developed in
\cite{xu2026generalized-rs}, where selected antenna processing units are electronically
activated as discrete and reconfigurable radiation points.

Existing blockage-aware and geometry-aware studies have established effective methods for exploiting environment information in link-level pinching-antenna design, typically optimizing antenna positions and transmission parameters for a given user set. These works provide an important foundation for environment-aware pinching-antenna operation. In this paper, we investigate a complementary geometry-aware network-level problem, where the objective is not tied to a particular instantaneous user realization. Instead, the known blockage geometry is used to characterize the spatial service capability over the entire region, and the pinching-antenna configuration is optimized according to region-level metrics such as coverage and worst-location average SNR.

Motivated by the aforementioned discussions, we investigate geometry-aware network-level pinching-antenna design methodologies in obstacle-rich environments under a uniform traffic pattern. The main contributions of Part II are summarized as follows:
\begin{itemize}
    \item {\bf Geometry-aware network-level modeling and metrics: }
     We develop a geometry-aware network-level modeling framework for pinching-antenna systems in blockage-rich environments. The service region is discretized into a grid set, and the propagation environment is represented by a collection of three-dimensional cuboid blockages. To capture geometry-induced LoS/NLoS transitions, we introduce a deterministic LoS visibility indicator for each candidate pinching-antenna position and each grid point, which specifies whether the corresponding LoS path is obstructed by any blockage. Under a discrete activation model, we then derive a tractable per-grid average received-SNR expression whose geometry-dependent terms can be precomputed offline once the blockage layout is given. Leveraging these per-grid quantities, we further define a geometry-aware network-level coverage metric and a worst-grid robustness metric, which serve as the basis for the subsequent network-level optimizations.
    \item {\bf Geometry-aware problem formulation and optimization: }
     Building on a blockage-aware, geometry-dependent per-grid average-SNR model, we formulate two network-level design problems: (i) an average-SNR-threshold coverage maximization that selects the activation matrix to maximize the number of grids whose average SNR exceeds a target, and (ii) a fairness-oriented worst-grid average-SNR maximization that improves the minimum average SNR over the service region. We first show that the coverage maximization problem is NP-hard via a polynomial-time reduction from the Maximum Coverage problem, and derive an equivalent mixed-integer linear programming reformulation by introducing binary coverage variables linked through linear SNR constraints. To obtain scalable solutions, we further develop a structure-exploiting coordinate-ascent method that updates one waveguide at a time using precomputable SNR contributions. For the worst-grid objective, we adopt an epigraph reformulation and leverage the resulting monotone feasibility in the target SNR level to develop a low-complexity bisection-based heuristic, where each feasibility search is implemented efficiently by iteratively updating one waveguide at a time to reduce the total SNR shortfall over all grids.
\end{itemize}
Extensive simulations validate the proposed geometry-aware designs and quantify their gains. The results show that optimized pinching-antenna activation consistently outperforms fixed-array and random baselines in both average-SNR-threshold coverage and worst-grid average SNR, with clear scaling trends versus the SNR threshold, transmit power, NLoS strength, and the numbers of waveguides/candidates. Top-view SNR maps further illustrate distinct objective-driven behaviors: coverage maximization exploits favorable LoS corridors, whereas worst-grid optimization yields a more balanced activation to mitigate geometry-induced dead zones.

The rest of the paper is organized as follows. Section \ref{sec:system_traffic_channel} presents the system model, the geometry-aware LoS/NLoS channel model with explicit blockage layouts, and the proposed network-level performance metrics together with their grid-based formulation. Section \ref{sec:cov_max} formulates the average-SNR-threshold coverage maximization problem and develops a low-complexity coordinate-ascent algorithm. Section \ref{sec:avg_snr_opt} considers a fairness-oriented worst-grid average-SNR maximization design and proposes a bisection-based algorithm with efficient feasibility search via coordinate updates. Section \ref{sec:summary} summarizes and compares the traffic-aware and geometry-aware pinching antenna systems and lists several promising future directions. Section \ref{sec:simulation} provides numerical results and performance discussions. Finally, Section \ref{sec: conclusion} concludes the paper.


\section{System Model, Geometry-Aware Channel, and Performance Metrics}
\label{sec:system_traffic_channel}

In this section, we present the downlink pinching-antenna system with explicit blockages, the geometry-aware channel and signal models for discrete pinching-antenna activation, as well as the network-level performance metrics used for the subsequent optimization.

\subsection{System Model with Blockages and Discrete Candidates}
\label{subsec:system_model}

Consider a rectangular communication area of size $D_x\times D_y$, as shown in Fig. \ref{fig: system model siso}(b). 
A CPU coordinates $N$ dielectric waveguides deployed below the ceiling at height $d_v$, all aligned along the $x$-axis and uniformly spaced along the $y$-axis with spacing $d_h = D_y/(N-1)$ satisfying $d_h \gg \lambda$, where $\lambda$ denotes the carrier wavelength. The index set of waveguides is $\mathcal{N} \triangleq \{1,\ldots,N\}$. The feed point position of the $n$-th waveguide is denoted by
$\boldsymbol{\psi}_{0,n} = \big[0,(n-1)d_h - D_y/2,d_v\big]^{\mathsf T}, \forall n\in\mathcal{N}$.

In this work, we adopt a  discrete pinching-antenna activation model,
which serves as a practical counterpart to the continuous positioning
model studied in Part~I. Specifically, on each waveguide we predefine a finite
set of $M$ feasible pinching-antenna locations, and the design problem is to
select exactly one active candidate per waveguide to optimize the network-level
performance.

Formally, for each waveguide $n\in\mathcal{N}$, we define the candidate index
set $\mathcal{M}\triangleq\{1,\ldots,M\}$ with $x$-coordinates
$\{\widetilde x_{n,1},\ldots,\widetilde x_{n,M}\}\subset[0,D_x]$. The three-dimensional
position of candidate $m\in\mathcal{M}$ on waveguide $n$ is given by
\begin{align}
    \boldsymbol{\psi}^{\mathrm{pin}}_{n,m}
    =
    \big[\widetilde x_{n,m},\,(n-1)d_h - D_y/2,\, d_v\big]^{\mathsf T}.
\end{align}
To indicate which candidate is activated on each waveguide, we introduce binary
variables $a_{n,m}\in\{0,1\}$, where $a_{n,m}=1$ means that the pinching antenna
at $\boldsymbol{\psi}^{\mathrm{pin}}_{n,m}$ is activated and $a_{n,m}=0$
otherwise. Our design activates exactly one pinching antenna per waveguide,
which yields
\begin{align}
    \sum_{m=1}^{M} a_{n,m} = 1,\ \forall n\in\mathcal{N}.
    \label{eq:row_sum}
\end{align}
Collecting $\{a_{n,m}\}$ into a binary activation matrix results in
\begin{align}
    \mathbf{A}
    \triangleq
    \big[a_{n,m}\big]_{n\in\mathcal{N},\,m\in\mathcal{M}}
    \in \{0,1\}^{N\times M},
\end{align}
where each row of $\mathbf{A}$ contains exactly one nonzero entry and specifies the
active pinching-antenna position on the corresponding waveguide. The matrix
$\mathbf{A}$ will serve as the optimization variable in the subsequent
geometry-aware design.

As in Part~I, we discretize the service region into $N_h$ grids along the
$x$-direction and $N_v$ grids along the $y$-direction. The grid sizes are written as
\begin{align}
    \Delta_u \triangleq \frac{D_x}{N_h},\ 
    \Delta_v \triangleq \frac{D_y}{N_v}.
\end{align}
The center of grid $(u,v)$ is given by
\begin{align}
    x_u = \Big(u-\tfrac{1}{2}\Big)\Delta_u, \ 
    y_v = -\tfrac{D_y}{2} + \Big(v-\tfrac{1}{2}\Big)\Delta_v,
\end{align}
for $u=1,\ldots,N_h$ and $v=1,\ldots,N_v$. The representative user position in
grid $(u,v)$ is given by
\begin{align}
    \boldsymbol{\psi}_{u,v} = [x_u, y_v, 0]^{\mathsf T},
\end{align}
and we denote the set of grid indices by
\begin{align}
    \Omega \triangleq
    \{1,\ldots,N_h\}\times\{1,\ldots,N_v\},\ 
    |\Omega| = N_h N_v .
\end{align}

This discrete-activation model offers two advantages in the considered
geometry-aware setting. First, it simplifies control and implementation by allowing each
waveguide to be equipped with a finite number of mechanical or electrical tap
points, and the BS only needs to select one tap per waveguide. The
configuration can thus be encoded by low-rate digital control signals (one
index per waveguide), without requiring continuously tunable positions or fine
mechanical actuation. Second, it turns continuous antenna placement into a
finite selection problem, which enables efficient geometry-aware optimization via
precomputable channel and coverage datasets over the grid set $\Omega$, as will be shown in the following subsections.

\subsection{Geometry-Aware Channel Model}
\label{subsec:channel_model}

We consider $K$ rectangular cuboid blockages (e.g.,
walls or partitions) in the service region. The three-dimensional geometry of
blockage $k\in\mathcal{K}\triangleq\{1,\ldots,K\}$ is modeled as
\begin{align}
    \mathcal{B}_k
    =
    [x_k^{\min}, x_k^{\max}]
    \times
    [y_k^{\min}, y_k^{\max}]
    \times
    [0,H_{\mathrm{blk}}],
\end{align}
where $H_{\mathrm{blk}}<d_v$ denotes the blockage height,
$x_k^{\min},x_k^{\max}\in[0,D_x]$, and
$y_k^{\min},y_k^{\max}\in[-D_y/2,D_y/2]$.

Recall that the service region is discretized into grids indexed by
$(u,v)\in\Omega$, with representative location
$\boldsymbol{\psi}_{u,v}=[x_u,y_v,0]^{\mathsf T}$. For grid $(u,v)$ and a
pinching-antenna candidate $(n,m)$, the complex baseband channel coefficient is
denoted by $h_{n,m}(\boldsymbol{\psi}_{u,v})\in\mathbb{C}$. We adopt a
geometry-aware LoS/NLoS model, where the LoS component is deterministically
determined by the blockage layout.

We introduce a deterministic LoS visibility indicator
$\chi_{n,m}(\boldsymbol{\psi}_{u,v})\in\{0,1\}$, where
$\chi_{n,m}(\boldsymbol{\psi}_{u,v})=1$ means that the LoS path between
candidate $(n,m)$ and grid $(u,v)$ is not obstructed by any blockage, and
$\chi_{n,m}(\boldsymbol{\psi}_{u,v})=0$ otherwise. Accordingly, the channel is
modeled as
\begin{equation}
    h_{n,m}(\boldsymbol{\psi}_{u,v})
    =
    \chi_{n,m}(\boldsymbol{\psi}_{u,v})\,
        h_{n,m}^{\mathrm{LoS}}(\boldsymbol{\psi}_{u,v})
    +
    h_{n,m}^{\mathrm{NLoS}}(\boldsymbol{\psi}_{u,v}),
    \label{eq:hn-decomposition-geo}
\end{equation}
where the term $h_{n,m}^{\mathrm{LoS}}(\boldsymbol{\psi}_{u,v})$ denotes the LoS
component when it exists, while
$h_{n,m}^{\mathrm{NLoS}}(\boldsymbol{\psi}_{u,v})$ represents the NLoS channel
component \cite{3gpp2020channel}.

Following the channel model in \cite{ding2024flexible,xu2025pinching}, the LoS
component can be written as
\begin{equation}
    h_{n,m}^{\mathrm{LoS}}(\boldsymbol{\psi}_{u,v})
    =
    \sqrt{\eta}\,
    \frac{\exp\!\Big(
        -j(\frac{2\pi}{\lambda}\,r_{n,m}(\boldsymbol{\psi}_{u,v})
        +\frac{2\pi}{\lambda_g}\,\widetilde x_{n,m})
    \Big)}
    {r_{n,m}(\boldsymbol{\psi}_{u,v})},
    \label{eq:hn-los-geo}
\end{equation}
where $\lambda$ is the carrier wavelength, $\lambda_g$ is the guided
wavelength, and $\eta=c^2/(4\pi f_c)^2$ with $c$ the speed of light and $f_c$
the carrier frequency. Moreover,
$r_{n,m}(\boldsymbol{\psi}_{u,v}) \triangleq \sqrt{(x_u-\widetilde x_{n,m})^2 + C_{n,v}}$ denotes the distance between grid $(u,v)$ with $C_{n,v}\triangleq \big(y_v - ((n-1)d_h - D_y/2)\big)^2 + d_v^2$.

The NLoS component $h_{n,m}^{\mathrm{NLoS}}(\boldsymbol{\psi}_{u,v})$ captures
the aggregate contribution of scattered multipath. We adopt a Rayleigh fading
model in which the NLoS term is the superposition of $N_c$ independent
scattering clusters
\begin{equation}
    h_{n,m}^{\mathrm{NLoS}}(\boldsymbol{\psi}_{u,v}) =
    \sum_{\ell=1}^{N_c} g_{n,m,\ell}(\boldsymbol{\psi}_{u,v}),
    \label{eq:nlos_channel}
\end{equation}
where $g_{n,m,\ell}(\boldsymbol{\psi}_{u,v})$ denotes the small-scale fading
coefficient associated with the $\ell$-th cluster. We model it as
\begin{equation}
    g_{n,m,\ell}(\boldsymbol{\psi}_{u,v})
    \sim
    \mathcal{CN}\!\left(
        0,\,
        \frac{\mu_\ell^2}{r_{n,m}^2(\boldsymbol{\psi}_{u,v})}
    \right),
    \label{eq:gnl-grid}
\end{equation}
so that each cluster has zero mean and a distance-dependent average power,
where $\mu_\ell^2$ is the average power of cluster $\ell$.

\begin{rem} 
The channel models adopted in Parts I and II correspond to different
levels of available environment information. In Part I, detailed blockage
geometry is assumed unavailable, and the LoS state is therefore modeled
statistically through a Bernoulli random variable with a
distance-dependent LoS probability. In contrast, Part II assumes that the
blockage layout is known, so the LoS state between each candidate
pinching-antenna location and each grid can be explicitly determined from
the environment map through the deterministic visibility indicator
$\chi_{n,m}(\boldsymbol{\psi}_{u,v})$. In both cases, the NLoS component
remains random to capture scattered multipath. Hence, Part I provides a
statistical blockage model for environments without detailed geometry
information, whereas Part II exploits explicit geometry information for
deterministic LoS identification.
\end{rem}

\subsection{Signal Model}
\label{subsec:signal_model}

Given an activation matrix $\mathbf{A}$, the
effective channel from waveguide $n$ to grid $(u,v)$ is given by
\begin{equation}
    \widetilde h_n(\boldsymbol{\psi}_{u,v};\mathbf{A})
    \triangleq
    \sum_{m=1}^{M} a_{n,m}\, h_{n,m}(\boldsymbol{\psi}_{u,v}),
    \ \forall n\in\mathcal{N},\ (u,v)\in\Omega,
    \label{eq:heff_def}
\end{equation}
which is essentially the channel coefficient of the activated candidate on
waveguide $n$ due to $\sum_{m=1}^{M} a_{n,m}=1$. Stacking all waveguides, we
define the effective channel vector as
\begin{equation}
    \widetilde{\boldsymbol{h}}(\boldsymbol{\psi}_{u,v};\mathbf{A})
    \triangleq
    \big[\widetilde h_1(\boldsymbol{\psi}_{u,v};\mathbf{A}),\ldots,
         \widetilde h_N(\boldsymbol{\psi}_{u,v};\mathbf{A})\big]^{\mathsf T}
    \in\mathbb{C}^{N}.
    \label{eq:heff_vec}
\end{equation}

Let $s$ denote the downlink data symbol, satisfying $\mathbb{E}[|s|^2]=1$. The received baseband signal at
grid $(u,v)$ is given by
\begin{align}
    y(\boldsymbol{\psi}_{u,v}) = \sqrt{P}\,
    \widetilde{\boldsymbol{h}}^{\mathsf T}(\boldsymbol{\psi}_{u,v};\mathbf{A})
    \boldsymbol{w}(\boldsymbol{\psi}_{u,v};\mathbf{A})\,s
    + z_{u,v},
    \label{eq:rx_signal}
\end{align}
where $P$ denotes the total transmit power, $\boldsymbol{w}(\boldsymbol{\psi}_{u,v};\mathbf{A}) \in \Cbb^N$ denotes the downlink beamformer, and $z_{u,v}\sim\mathcal{CN}(0,\sigma^2)$ denotes the received additive white Gaussian noise with $\sigma^2$ representing the noise power.

Similar to Part I, the BS employs maximum-ratio transmission (MRT) beamformer based on the effective channel
$\widetilde{\boldsymbol{h}}(\boldsymbol{\psi}_{u,v};\mathbf{A})$, i.e.,
\begin{equation}
    \boldsymbol{w}(\boldsymbol{\psi}_{u,v};\mathbf{A})
    \triangleq
    \frac{\widetilde{\boldsymbol{h}}^{*}(\boldsymbol{\psi}_{u,v};\mathbf{A})}
    {\big\|\widetilde{\boldsymbol{h}}(\boldsymbol{\psi}_{u,v};\mathbf{A})\big\|},
    \
    \big\|\boldsymbol{w}(\boldsymbol{\psi}_{u,v};\mathbf{A})\big\|^2=1,
    \label{eq:mrt_precoder}
\end{equation}
where $(\cdot)^*$ denotes complex conjugation. 
Under MRT, the instantaneous received SNR at grid $(u,v)$ becomes
\begin{align}
    \Gamma(\boldsymbol{\psi}_{u,v};\mathbf{A}) &\triangleq
    \frac{P\Big| \widetilde{\boldsymbol{h}}^{\mathsf T}(\boldsymbol{\psi}_{u,v};\mathbf{A})
    \boldsymbol{w}(\boldsymbol{\psi}_{u,v};\mathbf{A})
    \Big|^2}{\sigma^2} \notag\\
    &= \rho \big\|\widetilde{\boldsymbol{h}}(\boldsymbol{\psi}_{u,v};\mathbf{A})\big\|^2,
    \label{eq:snr_mrt}
\end{align}
where $\rho \triangleq P/\sigma^2$.

\subsection{Per-Grid Average Received SNR}
\label{subsec:avg_snr_discussion}

Under the blockage-aware channel model, the received SNR varies due to
small-scale fading (mainly through the NLoS component). We therefore adopt an
average-SNR-threshold performance metric. 
For the per-grid average received SNR, we have the following lemma
\begin{lem} \label{lem: per-grid snr} 
    Given an activation matrix $\mathbf{A}$, for grid $(u,v)\in\Omega$, the per-grid average received SNR is given by
    \begin{align}
    \bar{\Gamma}(\boldsymbol{\psi}_{u,v};\mathbf{A}) =
        \rho \sum_{n=1}^{N}\sum_{m=1}^{M}
        a_{n,m}\,\bar{G}_{n,m}(\boldsymbol{\psi}_{u,v})
        \label{eq:avg_snr_final}
    \end{align}
    where
    $\bar{G}_{n,m}(\boldsymbol{\psi}_{u,v})
    \triangleq
    \big(\chi_{n,m}(\boldsymbol{\psi}_{u,v})\eta+\mu^2\big)/
    r_{n,m}^2(\boldsymbol{\psi}_{u,v})$
    denotes the average channel power from candidate $(n,m)$ to grid $(u,v)$.
\end{lem}

{\emph{Proof:}} See Appendix \ref{appd: per-grid snr}.  \hfill$\blacksquare$

Lemma \ref{lem: per-grid snr} shows that, for each grid $(u,v)$, the average
SNR under activation $\mathbf{A}$ depends on $\{a_{n,m}\}$ only through a
weighted sum. The weights are entirely determined by the environment geometry
(via $\chi_{n,m}(\boldsymbol{\psi}_{u,v})$ and $r_{n,m}(\boldsymbol{\psi}_{u,v})$)
and the channel parameters $(\eta,\mu^2)$. Consequently, once the blockage
layout and the grid are fixed, the terms
$\big(\chi_{n,m}(\boldsymbol{\psi}_{u,v})\,\eta+\mu^2\big)
/r_{n,m}^2(\boldsymbol{\psi}_{u,v})$ can be precomputed offline for all
$(u,v)\in\Omega$, $n\in\mathcal{N}$, and $m\in\mathcal{M}$. These precomputed
SNR maps will be used in the following sections to construct geometry-aware coverage
metrics and formulate the corresponding network optimization problems.

We emphasize that the linear dependence on $\mathbf{A}$ in \eqref{eq:avg_snr_final}
stems from the discrete-activation model, where each waveguide selects exactly
one candidate from a finite set. This structure supports an
offline--online workflow. In particular, once the environment is specified, the
geometry-dependent channel/SNR terms can be precomputed offline, and the
online design reduces to selecting the activation matrix
$\mathbf{A}\in\{0,1\}^{N\times M}$. In what follows, we build on
\eqref{eq:avg_snr_final} to formulate geometry-aware pinching-antenna activation
optimization problems that choose $\mathbf{A}$ to best match the blockage layout
and the desired coverage objectives over the grid set $\Omega$. In particular,
we focus on average-SNR-threshold coverage maximization and worst-grid
average-SNR maximization.

\section{Average-SNR-Threshold Coverage Maximization}
\label{sec:cov_max}

In this section, we study a coverage-oriented pinching-antenna activation
design for the geometry-aware scenario. The key idea is to map the per-grid
average SNR in \eqref{eq:avg_snr_final} to a binary notion of coverage via an
SNR target, and then select the activation matrix $\mathbf{A}$ to maximize the
resulting covered area over the grid set $\Omega$. To this end, we first
introduce the following definition.

\subsection{ Problem Formulation}

\begin{definition}[SNR-threshold coverage]
    Given an SNR threshold $\gamma_{\mathrm{th}}>0$, grid $(u,v)$ is said to be
    \emph{covered} under activation $\mathbf{A}$ if its average received SNR satisfies
    $\bar{\Gamma}(\boldsymbol{\psi}_{u,v};\mathbf{A})\ge \gamma_{\mathrm{th}}$.
    Accordingly, we define the binary coverage indicator as
    \begin{equation}
        c_{u,v}(\mathbf{A})
        \triangleq
        \mathbbm{1}\!\left\{
            \bar{\Gamma}(\boldsymbol{\psi}_{u,v};\mathbf{A})
            \ge \gamma_{\mathrm{th}}
        \right\},
        \ (u,v)\in\Omega,
        \label{eq:coverage_indicator}
    \end{equation}
    where $\mathbbm{1}\{\cdot\}$ denotes the indicator function.
\end{definition}

With the above definition, we can now quantify the covered area over the grid
set $\Omega$ by counting the number of covered grids. Accordingly, the
average-SNR-threshold coverage maximization problem is formulated as
\begin{subequations}
\begin{align}
    \max_{\mathbf{A}}
    \quad &
    \sum_{(u,v)\in\Omega} c_{u,v}(\mathbf{A}) \label{eqn:cov_max_obj}\\
    \text{s.t.}\quad &
    \sum_{m=1}^{M} a_{n,m} = 1,\ \forall n\in\mathcal{N},\\
    &
    a_{n,m}\in\{0,1\},\ \forall n\in\mathcal{N},\ m\in\mathcal{M}.
\end{align}
\label{prob:cov_max}
\end{subequations}
Here, the objective maximizes the total number of grids whose average received
SNR exceeds the target $\gamma_{\mathrm{th}}$, while the constraints enforce
the one-hot activation of a single candidate on each waveguide.

It is challenging to solve the Problem~\eqref{prob:cov_max} because the objective involves the
discontinuous indicator $c_{u,v}(\mathbf{A})$. Fortunately, thanks to the
linear form of $\bar{\Gamma}(\boldsymbol{\psi}_{u,v};\mathbf{A})$ in
\eqref{eq:avg_snr_final} and the one-hot activation constraints, the problem
admits a clean mixed-integer linear programming (MILP) reformulation, as well
as a scalable structure-exploiting algorithm.

\subsection{Exact MILP Reformulation}
To handle the discontinuous coverage indicator, we introduce auxiliary binary
variables $\{c_{u,v}\in\{0,1\}\}_{(u,v)\in\Omega}$, where $c_{u,v}=1$ indicates
that grid $(u,v)$ is covered. The coverage condition
$\bar{\Gamma}(\boldsymbol{\psi}_{u,v};\mathbf{A})\ge \gamma_{\mathrm{th}}$ can be
linked to $c_{u,v}$ through the implication
\[
c_{u,v}=1 \ \Rightarrow\  \bar{\Gamma}(\boldsymbol{\psi}_{u,v};\mathbf{A})\ge\gamma_{\mathrm{th}},
\]
which is enforced by the linear constraint
\begin{equation}
    \bar{\Gamma}(\boldsymbol{\psi}_{u,v};\mathbf{A})
    \ge
    \gamma_{\mathrm{th}}\,c_{u,v},\ \forall (u,v)\in\Omega.
    \label{eq:cov_milp_link}
\end{equation}
Importantly, since $\bar{\Gamma}(\boldsymbol{\psi}_{u,v};\mathbf{A}) \geq 0$, \eqref{eq:cov_milp_link} does not require any big-$M$ constant.
The following lemma shows that the resulting formulation is equivalent to the
original coverage-count objective.

\begin{lem}
For any fixed activation matrix $\mathbf{A}$, an optimal choice of
$\{c_{u,v}\}$ under \eqref{eq:cov_milp_link} is given by
\begin{equation}
    c_{u,v}^\star
    =
    \mathbbm{1}\!\left\{
        \bar{\Gamma}(\boldsymbol{\psi}_{u,v};\mathbf{A})\ge\gamma_{\mathrm{th}}
    \right\},
    \ \forall (u,v)\in\Omega.
\end{equation}
Consequently, maximizing $\sum_{(u,v)\in\Omega} c_{u,v}$ subject to
\eqref{eq:cov_milp_link} is equivalent to maximizing the covered-grid count
$\sum_{(u,v)\in\Omega}\mathbbm{1}\{\bar{\Gamma}(\boldsymbol{\psi}_{u,v};\mathbf{A})\ge\gamma_{\mathrm{th}}\}$.
\end{lem}


With the introduced auxiliary variables $\{c_{u,v}\}$ and the linear linking
constraints \eqref{eq:cov_milp_link}, Problem \eqref{prob:cov_max} is
equivalently reformulated as the following MILP
\begin{subequations}
\begin{align}
    \max_{\mathbf{A},\,\{c_{u,v}\}}
    \quad &
    \sum_{(u,v)\in\Omega} c_{u,v} \label{eq:23a}\\
    \text{s.t.}\quad &
    \bar{\Gamma}(\boldsymbol{\psi}_{u,v};\mathbf{A})
    \ge
    \gamma_{\mathrm{th}}\,c_{u,v},
    \ \forall (u,v)\in\Omega, \label{eq:23b}\\
    &
    \sum_{m=1}^{M} a_{n,m} = 1,\ \forall n\in\mathcal{N}, \label{eq:23c} \\
    &
    a_{n,m}\in\{0,1\},\ \forall n\in\mathcal{N},\ m\in\mathcal{M}, \label{eq:23d} \\
    &
    c_{u,v}\in\{0,1\},\ \forall (u,v)\in\Omega, \label{eq:23e}
\end{align}
\label{prob:cov_milp}
\end{subequations} %
Compared with the original formulation \eqref{prob:cov_max}, the MILP reformulation
\eqref{prob:cov_milp} removes the discontinuous indicator function from the objective by
introducing auxiliary binary coverage variables and linear linking constraints. As a result,
\eqref{prob:cov_milp} becomes a standard mixed-integer linear program that can be solved by
off-the-shelf MILP solvers when the grid size $|\Omega|$ and the candidate count $M$ are
moderate. Nevertheless, as stated in Proposition~\ref{prop:NP_hard_MILP_coverage},
problem \eqref{prob:cov_milp} is still NP-hard in general and may incur prohibitive complexity
when $|\Omega|$ and/or $M$ is large.

\begin{prop}
\label{prop:NP_hard_MILP_coverage}
Problem \eqref{prob:cov_milp} is NP-hard.
\end{prop}

\emph{Proof:} See Appendix \ref{appd:NP_hard_MILP_coverage}. 
\hfill $\blacksquare$

The NP-hardness of \eqref{prob:cov_milp} implies that globally solving the problem may become computationally prohibitive for large network instances. This motivates us to devise a low-complexity algorithm to obtain high-quality solutions with scalable computational cost.

\subsection{A Structure-Exploiting Coordinate-Ascent Algorithm}
\label{subsec:cov_coord_algo}

While Problem~\eqref{prob:cov_milp} can be solved by off-the-shelf MILP solvers
(e.g., via branch-and-bound) for moderate problem sizes, its computational
burden may become prohibitive when the grid set $\Omega$ is dense or the number
of candidates $M$ is large \cite{wolsey1999integer}. We therefore
develop a low-complexity algorithm that exploits two key structures: (i) the
one-hot constraint $\sum_{m=1}^{M}a_{n,m}=1$ for each waveguide $n$, and (ii)
the additive form of the per-grid average SNR in \eqref{eq:avg_snr_final}.

Given an activation matrix $\mathbf{A}$, let $m_n$ denote the
currently activated candidate on waveguide $n$, i.e., $a_{n,m_n}=1$. For each
grid $(u,v)$, define the average SNR excluding waveguide $n$ as
\begin{equation}
    \bar{\Gamma}_{-n}(\boldsymbol{\psi}_{u,v};\mathbf{A})
    \triangleq
    \bar{\Gamma}(\boldsymbol{\psi}_{u,v};\mathbf{A})
    - \rho\,\bar{G}_{n,m_n}(\boldsymbol{\psi}_{u,v}),
    \ (u,v)\in\Omega.
    \label{eq:residual_snr_cov_nosymbol}
\end{equation}
If waveguide $n$ switches to candidate $m$, the resulting average SNR at grid
$(u,v)$ becomes
\begin{align}
    \bar{\Gamma}(\boldsymbol{\psi}_{u,v};\,\mathbf{A}_{n\to m})
    &= \bar{\Gamma}_{-n}(\boldsymbol{\psi}_{u,v};\mathbf{A})\notag\\
    &\quad + \rho\,\bar{G}_{n,m}(\boldsymbol{\psi}_{u,v}),
    \ (u,v)\in\Omega,
    \label{eq:snr_after_switch}
\end{align}
where $\mathbf{A}_{n\to m}$ denotes the activation matrix obtained by replacing
row $n$ of $\mathbf{A}$ with $a_{n,m}=1$.

Therefore, when updating waveguide $n$, we can enumerate $m\in\mathcal{M}$ and
select the candidate that yields the largest number of covered grids
\begin{small}\begin{align}
    m_n^\star
    \in
    \arg\max_{m\in\mathcal{M}}
    \sum_{(u,v)\in\Omega}
    \mathbbm{1}\!\left\{
        \bar{\Gamma}_{-n}(\boldsymbol{\psi}_{u,v};\mathbf{A})
        + \rho\,\bar{G}_{n,m}(\boldsymbol{\psi}_{u,v})
        \ge \gamma_{\mathrm{th}}
    \right\}.
    \label{eq:best_m_cov_update_nosymbol}
\end{align}
\end{small} %

\noindent To improve stability, when multiple candidates achieve the same covered-grid
count, we use the following tie-breaker: choose the candidate that maximizes
the total SNR margin above the threshold,
\begin{equation}
    \sum_{(u,v)\in\Omega}
    \Big[
        \bar{\Gamma}_{-n}(\boldsymbol{\psi}_{u,v};\mathbf{A})
        + \rho\,\bar{G}_{n,m}(\boldsymbol{\psi}_{u,v})
        - \gamma_{\mathrm{th}}
    \Big]^+,
    \label{eq:tiebreak_margin_nosymbol}
\end{equation}
where $[x]^+\triangleq\max\{x,0\}$.

Starting from any feasible $\mathbf{A}^{(0)}$, we iterate over $n=1,\ldots,N$ and
update one waveguide at a time using \eqref{eq:best_m_cov_update_nosymbol}
(until no further improvement is observed). Each update optimizes the coverage
objective over the finite candidate set of the current waveguide while fixing
all other waveguides, and hence the covered-grid count is nondecreasing across
iterations. Since the feasible set is finite, the algorithm terminates after a
finite number of iterations.


Each waveguide update evaluates $M$ candidates over $|\Omega|$ grids, leading to
a per-iteration complexity of $\mathcal{O}(NM|\Omega|)$. Since the quantities
$\{\bar{G}_{n,m}(\boldsymbol{\psi}_{u,v})\}$ are precomputable offline, the
online computation consists of simple additions and threshold comparisons.

We note that, as a coordinate-wise method for a discrete nonconvex
problem, the resulting solution may depend on the initial activation
matrix. This sensitivity can be alleviated through a multi-start
implementation, where the algorithm is executed from several feasible
initializations and the best converged solution is retained.

\section{Worst-Grid Average-SNR Maximization}
\label{sec:avg_snr_opt}

In this section, we consider a fairness-oriented design objective. Specifically,
we seek to enhance the weakest grid in the service region by maximizing the
minimum per-grid average received SNR over $\Omega$. This criterion is
particularly relevant in blockage-aware environments, where some grids may
otherwise experience persistently low SNR due to unfavorable LoS visibility.

\subsection{Problem Formulation}
\label{subsec:wg_formulation}

Recall that $\bar{\Gamma}(\boldsymbol{\psi}_{u,v};\mathbf{A})$ denotes the
average received SNR at grid $(u,v)$ under activation $\mathbf{A}$. The
worst-grid average-SNR maximization problem can be formulated as
\begin{subequations}
    \begin{align}
        \max_{\mathbf{A}}
        \quad &
        \min_{(u,v)\in\Omega}\ \bar{\Gamma}(\boldsymbol{\psi}_{u,v};\mathbf{A}) \\
        \text{s.t.}\quad &
        \sum_{m=1}^{M} a_{n,m} = 1,\ \forall n\in\mathcal{N},\\
        &
        a_{n,m}\in\{0,1\},\ \forall n\in\mathcal{N},\ m\in\mathcal{M}.
    \end{align}\label{prob:wg_orig}
\end{subequations}%
To obtain an equivalent mixed-integer linear formulation, we introduce an
auxiliary variable $t$ representing the worst-grid average SNR. Then, problem
\eqref{prob:wg_orig} can be equivalently rewritten in epigraph form as
\begin{subequations}
    \begin{align}
        \max_{\mathbf{A},\,t}
        \quad & t \\
        \text{s.t.}\quad &
        \bar{\Gamma}(\boldsymbol{\psi}_{u,v};\mathbf{A}) \ge t,
        \qquad \forall (u,v)\in\Omega,\\
        & \sum_{m=1}^{M} a_{n,m} = 1,\ \forall n\in\mathcal{N},\\
        & a_{n,m}\in\{0,1\},\ \forall n\in\mathcal{N},\ m\in\mathcal{M}.
    \end{align}
    \label{prob:wg_mip}
\end{subequations}
Although \eqref{prob:wg_mip} removes the inner minimization operator, it remains
challenging because it is an MILP involving binary variables
$\mathbf{A}$ and a continuous SNR variable $t$. To develop a structured solution,
we next exploit a monotonic feasibility property with respect to the target
average SNR level.

\subsection{Monotonic Feasibility in the Target SNR Level}
\label{subsec:wg_mono}

For a given target $t \ge 0$, consider the following feasibility question:
\emph{Is there an activation matrix $\mathbf{A}$ such that
$\bar{\Gamma}(\boldsymbol{\psi}_{u,v};\mathbf{A})\ge t$ holds for all
$(u,v)\in\Omega$?}
This feasibility question is monotone in $t$, as stated
below.

\begin{lem} \label{lem:mono_feas}
Given $t$, the feasible set is defined as
\begin{align}
    \mathcal{F}(t) \triangleq
    \Big\{ \mathbf{A}\in\{0,1\}^{N\times M}  \big| &
    \bar{\Gamma}(\boldsymbol{\psi}_{u,v};\mathbf{A})\ge t, \forall (u,v)\in\Omega,\notag\\
    &\sum_{m=1}^{M}a_{n,m}=1, \forall n\in\mathcal{N}
    \Big\}.
\end{align}
For any $0\le t_1 \le t_2$, we have $\mathcal{F}(t_2)\subseteq\mathcal{F}(t_1)$.
Equivalently, if the constraints are feasible at $t_2$, then they are also
feasible at any $t_1\in[0,t_2]$.
\end{lem}

\emph{Proof:}
Let $0\le t_1 \le t_2$ and take any $\mathbf{A}\in\mathcal{F}(t_2)$.
Then $\bar{\Gamma}(\boldsymbol{\psi}_{u,v};\mathbf{A})\ge t_2$ for all
$(u,v)\in\Omega$, and the one-hot constraints $\sum_{m=1}^{M}a_{n,m}=1$ hold for
all $n\in\mathcal{N}$. Since $t_1\le t_2$, it follows that
$\bar{\Gamma}(\boldsymbol{\psi}_{u,v};\mathbf{A})\ge t_1$ for all
$(u,v)\in\Omega$, while the one-hot constraints remain unchanged. Hence
$\mathbf{A}\in\mathcal{F}(t_1)$, which proves
$\mathcal{F}(t_2)\subseteq\mathcal{F}(t_1)$.
\hfill$\blacksquare$

Lemma~\ref{lem:mono_feas} shows that the feasibility of the per-grid constraints
$\bar{\Gamma}(\boldsymbol{\psi}_{u,v};\mathbf{A})\ge t$ is monotone in $t$:
whenever a target level $t$ is achievable, any smaller level is also
achievable. This monotonic structure implies that the set of achievable
worst-grid SNR levels forms an interval and naturally motivates applying a
bisection search over $t$ to solve Problem~\eqref{prob:wg_mip}.

\subsection{Low-Complexity Bisection-Based Algorithm for Solving Problem~\eqref{prob:wg_mip}}
\label{subsec:wg_bisect_heur}

The monotonic feasibility property in
Lemma~\ref{lem:mono_feas} naturally motivates
searching over the target SNR level $t$. If the feasibility problem in
\eqref{prob:wg_mip} is solved exactly at each tested value of $t$, a standard bisection
procedure can recover the globally optimal worst-grid average SNR up to
the prescribed numerical tolerance. However, repeatedly solving the
binary feasibility problem can incur high computational complexity as
the numbers of waveguides, candidate locations, and spatial grids
increase. To obtain a scalable solution, we instead develop a
low-complexity bisection-based method in which the feasibility search at
each tested $t$ is carried out through coordinate-wise deficit
reduction.

\subsubsection{Feasibility Search at a Given SNR Level $t$}

For a fixed test value $t$, determining whether the target is globally
feasible amounts to finding an activation matrix $\mathbf{A}$ satisfying
\begin{subequations} \label{eqn:local feasibility problem}
    \begin{align}
    \text{find}\quad & \mathbf{A} \\
    \text{s.t.}\quad
    & \bar{\Gamma}(\boldsymbol{\psi}_{u,v};\mathbf{A}) \geq t,
    \quad \forall (u,v)\in\Omega, \\
    & \sum_{m=1}^{M}a_{n,m}=1,\quad \forall n\in\mathcal{N},  \\
    & a_{n,m}\in\{0,1\},\quad \forall n\in\mathcal{N},~m\in\mathcal{M}.
    \end{align}
\end{subequations}

Problem~\eqref{eqn:local feasibility problem} is a binary linear feasibility problem and can be solved
exactly using a mixed-integer solver. However, repeatedly solving such
problems within the search over $t$ can be computationally expensive.
We therefore develop a low-complexity coordinate-wise procedure to
search for a feasible activation matrix at each tested SNR level.

\subsubsection{Low-Complexity Feasibility Search via Deficit Reduction}
For a fixed level $t$ and a current one-hot activation matrix $\mathbf{A}$,
define the per-grid SNR deficit
\begin{equation}
    d_{u,v}(\mathbf{A};t)
    \triangleq
    \Big[t-\bar{\Gamma}(\boldsymbol{\psi}_{u,v};\mathbf{A})\Big]^+,
    \qquad (u,v)\in\Omega,
\end{equation}
and the total deficit
\begin{equation}
    D(\mathbf{A};t)
    \triangleq
    \sum_{(u,v)\in\Omega} d_{u,v}(\mathbf{A};t).
\end{equation}
Clearly, $D(\mathbf{A};t)=0$ implies that $\mathbf{A}$ satisfies
$\bar{\Gamma}(\boldsymbol{\psi}_{u,v};\mathbf{A})\ge t$ for all $(u,v)\in\Omega$.

Starting from an initial one-hot $\mathbf{A}$, we sweep over waveguides
$n=1,\ldots,N$ and update one waveguide at a time to reduce $D(\mathbf{A};t)$.
Let $m_n$ denote the currently activated candidate on waveguide $n$ (i.e.,
$a_{n,m_n}=1$). As in \eqref{eq:residual_snr_cov_nosymbol}, define
\begin{equation}
    \bar{\Gamma}_{-n}(\boldsymbol{\psi}_{u,v};\mathbf{A})
    \triangleq
    \bar{\Gamma}(\boldsymbol{\psi}_{u,v};\mathbf{A})
    - \rho\,\bar{G}_{n,m_n}(\boldsymbol{\psi}_{u,v}),
    \ (u,v)\in\Omega.
\end{equation}
If waveguide $n$ switches to candidate $m$, the resulting per-grid average SNR
is given by \eqref{eq:snr_after_switch}, i.e.,
$\bar{\Gamma}(\boldsymbol{\psi}_{u,v};\mathbf{A}_{n\to m})
=\bar{\Gamma}_{-n}(\boldsymbol{\psi}_{u,v};\mathbf{A})
+\rho\,\bar{G}_{n,m}(\boldsymbol{\psi}_{u,v})$.
We then choose the candidate that minimizes the updated total deficit:
\begin{equation}
    m_n^\star \in \arg\min_{m\in\mathcal{M}}
    \sum_{(u,v)\in\Omega}
    \Big[t-\bar{\Gamma}(\boldsymbol{\psi}_{u,v};\mathbf{A}_{n\to m})\Big]^+.
    \label{eq:deficit_update_rule}
\end{equation}
After several sweeps, or once $D(\mathbf{A};t)$ no longer decreases,
the coordinate-wise search terminates. If $D(\mathbf{A};t)=0$, the
obtained activation matrix satisfies all the per-grid SNR constraints
in \eqref{eqn:local feasibility problem}, and hence the tested SNR level $t$ is feasible. In contrast,
if $D(\mathbf{A};t)>0$, this only indicates that the coordinate-wise
search fails to find a feasible activation matrix from the considered
initialization; it does not imply that Problem~\eqref{eqn:local feasibility problem} is globally
infeasible, since the discrete coordinate updates may converge to a
coordinate-wise local solution.

Accordingly, when this coordinate-wise feasibility search is embedded
into the bisection procedure, the resulting method should be regarded
as a low-complexity heuristic for Problem~\eqref{prob:wg_mip}, rather than a globally
optimal solver. Its practical benefit lies in replacing repeated global
mixed-integer optimization with simple candidate evaluations based on
the precomputed SNR maps. The resulting solution-quality--complexity
tradeoff will be evaluated against the globally optimal branch-and-bound (BnB)
solution of Problem~\eqref{prob:wg_mip} in Sec. \ref{sec:simulation}.

We also note that the coordinate-wise deficit-reduction search may also depend on the
initial activation matrix. In practice, a multi-start implementation can
be employed to alleviate this sensitivity by performing the search from
multiple feasible activation matrices and retaining the best obtained
solution.


\begin{table*}[t]
\centering
\caption{Summary and comparison between Part~I and Part~II .}
\label{tab:partI_partII_compare}
\renewcommand{\arraystretch}{1.48}
\begin{tabular}{p{2.6cm} p{7cm} p{7.2cm}}
\hline\hline
{\cellcolor{gray!26}\textbf{Aspect}} & {\cellcolor{gray!26}\textbf{Part~I: Traffic-aware case}} & {\cellcolor{gray!26}\textbf{Part~II: Geometry-aware case}} \\
\hline

\multirow{1.6}{*}{\textbf{Environment model}}
& Known traffic distribution modeled by a Gaussian-hotspot mixture; propagation follows a distance-dependent stochastic LoS/NLoS model (without explicit obstacles).
& Known blockage layout modeled by rectangular cuboids; LoS visibility is deterministically determined by geometry, while NLoS follows distance-dependent random fading.
\\

\multirow{1.6}{*}{\cellcolor{gray!16}\textbf{Service region model}}
& {\cellcolor{gray!16} Grid-based service region  to enable network-level performance evaluation/optimization.}
& {\cellcolor{gray!16}Grid-based model to enable geometry-aware coverage optimization and offline map construction.}
\\

\multirow{1.6}{*}{\textbf{Antenna control model}}
& Pinching antennas can move continuously along each waveguide; locations are real-valued design variables.
& Each waveguide has $M$ predefined candidate tap points; antenna activation is controlled via binary variables. \\

\multirow{1.6}{*}{\cellcolor{gray!16}\textbf{Performance metrics}}
& {\cellcolor{gray!16}(i) Traffic-weighted network average SNR; (ii) traffic-restricted worst-grid average SNR over an active set $\Omega_{\rm act}$.}
& {\cellcolor{gray!16}(i) Average-SNR-threshold coverage over $\Omega$; (ii) worst-grid average SNR over $\Omega$.} \\

\multirow{1.6}{*}{\textbf{Design objectives}}
& Continuous placement to optimize traffic-aware network average SNR and traffic-restricted worst-grid average-SNR.
& Discrete activation to optimize average-SNR-threshold coverage and worst-grid average-SNR. \\

\multirow{3.2}{*}{\cellcolor{gray!16}\textbf{Workflow and solver}}
& {\cellcolor{gray!16}Online continuous optimization: candidate-based low-complexity method for the traffic-weighted average SNR maximization, and BCD with bisection for the traffic-restricted worst-grid SNR maximization.}
& {\cellcolor{gray!16}Offline--online discrete optimization: offline SNR-map construction, followed by MILP/coordinate-ascent for coverage maximization and epigraph-based bisection with low-complexity feasibility search for worst-grid SNR optimization.}
\\

\hline\hline
\end{tabular}
\end{table*}

\section{Summary and Outlook}
\label{sec:summary}

Our two-part papers develop environment-aware design methodologies for
pinching-antenna systems from two complementary perspectives. In Part~I, we
consider a traffic-aware setting where user demand is modeled by a spatial
traffic distribution and the propagation environment is captured through a
distance-dependent stochastic LoS/NLoS model. Under this setting, we design
continuous pinching-antenna positions along the waveguides to optimize
traffic-aware network performance metrics. In Part~II, we shift to a
geometry-aware setting with explicit blockage layouts, where the LoS visibility
is deterministically determined by obstacles and each waveguide admits a finite
set of candidate pinching-antenna locations. This leads to a practical
discrete-activation formulation, for which we develop geometry-aware coverage
and fairness optimization strategies based on offline-precomputable SNR maps.
Table~\ref{tab:partI_partII_compare} summarizes the two baseline settings and
highlights how their distinct environment models and control capabilities lead
to different performance objectives and algorithmic solutions, together
providing a modular foundation for more sophisticated pinching-antenna system
designs.

Overall, this two-part work establishes two basic environment-aware
baselines for pinching-antenna network optimization. These two settings
are considered separately because the present studies are intended as
initial investigations of environment-aware network-level
pinching-antenna design. We therefore first establish and understand two
fundamental sources of spatial heterogeneity separately, namely traffic
heterogeneity and geometry-induced propagation heterogeneity. Part I
isolates the former and investigates how long-term spatial traffic demand
affects network-level pinching-antenna positioning, whereas Part II
focuses on the latter and exploits explicit environment geometry to
mitigate blockage-induced coverage degradation. Such a separation allows
the respective physical effects and optimization structures associated
with the two types of environment information to be identified more
clearly, thereby providing a methodological foundation for more general
joint designs.

The traffic-aware and geometry-aware settings are not mutually exclusive. When both the spatial traffic distribution and detailed environment geometry are available, they can be incorporated jointly within the same grid-based network-level framework. For example, letting $p_{u,v}$ denote the traffic probability of grid $(u,v)$, the geometry-aware coverage objective in (20) can be extended to the traffic-weighted coverage $\sum_{(u,v)\in\Omega}p_{u,v}c_{u,v}$, while retaining the same geometry-dependent SNR constraints and activation model. Similarly, geometry-aware per-grid SNRs can be incorporated into the traffic-weighted average-SNR and fairness metrics considered in Part I. Such a joint design, however, introduces additional challenges, including the tradeoff between serving high-demand regions and maintaining region-wide robustness. Moreover, when continuous pinching-antenna positioning is combined with deterministic geometry-aware blockage, the LoS state may change discontinuously across blockage boundaries, and thus the structural properties exploited in Part I may no longer directly apply. Joint traffic- and geometry-aware optimization therefore represents a natural and important extension of the present framework.

\begin{rem} 
The present model adopts cuboid obstacles, binary LoS visibility, and
uniform traffic as tractable abstractions for this initial geometry-aware
study. More irregular obstacle shapes and propagation effects such as
partial blockage or diffraction can be incorporated through more detailed
geometry/channel models when constructing the candidate-specific SNR maps,
while nonuniform traffic can be handled through the joint traffic- and
geometry-aware extension. Geometry uncertainty can
be addressed practically by enlarging obstacle regions or considering
multiple plausible blockage maps, and more fundamentally through
stochastic, chance-constrained, or worst-case robust optimization. Such
uncertainty-aware designs are more challenging due to the discontinuous
LoS states and deserve dedicated future investigation.
\end{rem}

\begin{rem} 
Although the current formulation focuses on evaluating one scheduled user at each
possible spatial location, the resulting pinching-antenna configuration
is intended for long-term regional service rather than a particular
instantaneous user. In a practical multi-user network, this long-term
configuration can be combined with short-term optimization of user
scheduling, resource allocation, beamforming/precoding, and power
allocation based on instantaneous user CSI. More generally, a
two-timescale joint design can further optimize the long-term antenna
configuration while accounting for the resulting multi-user transmission
performance, which deserves dedicated investigation in future work.
\end{rem}

\begin{table}[!t]
\centering
\caption{Default simulation parameters.}
\renewcommand{\arraystretch}{1.25}
\begin{tabular}{p{4.8cm} p{3.2cm}}
\hline \hline
\textbf{Parameter} & \textbf{Value} \\
\hline
Carrier frequency ($f_c$) & $28$ GHz \\
Communication-region size ($D_x \times D_y$) & $200 \times 60$ m$^2$ \\
Waveguide height ($d_v$) & $10$ m \\
Transmit power ($P$) & $40$ dBm \\
Noise power ($\sigma^2$) & $-70$ dBm \\
NLoS power ($\mu^2$) &  $-60$ dBm \\
Number of PAs / waveguides ($N$) & $4$ \\
Number of grids ($N_h \times N_v$) & $400 \times 120$ \\
Algorithm tolerance ($\epsilon_t,\epsilon_d$) & $10^{-3}$ \\
Obstacle height ($H_{\mathrm{blk}}$) & $6$ m \\ 
\hline \hline
\end{tabular}
\label{tab:sim-param}
\end{table}

\section{Simulation Results} \label{sec:simulation}

In this section, numerical results are presented to validate the performance of the proposed algorithms for geometry-aware pinching-antenna system designs. Unless otherwise stated, the main simulation parameters are summarized in Table~\ref{tab:sim-param}.
For the geometry-aware evaluation, users are assumed to be uniformly distributed over the communication region. All performance metrics are evaluated over a uniform $N_h \times N_v$ grid, after excluding the grid points that fall inside the obstacle footprints. The environment contains $K=6$ rectangular cuboid obstacles with height $H_{\mathrm{blk}}<d_v$. A direct link between a pinching antenna at height $d_v$ and a floor user is declared blocked if the corresponding three-dimensional line segment intersects the volume of any obstacle. Unless otherwise specified, the obstacles are placed as follows:
\begin{align}
    \mathcal{B}(x_c,y_1,y_2) \triangleq [x_c-\tfrac{w}{2},x_c+\tfrac{w}{2}]\times[y_1,y_2] \times [0, H_{\textrm{blk}}], 
\end{align}
where $x_c \in \{14,40,66\}, w=8$ m and $ (y_1,y_2) \in \{(0,20),(35,50)\}$ m.

\begin{figure}[!t]
	\centering
	\includegraphics[width=0.76\linewidth]{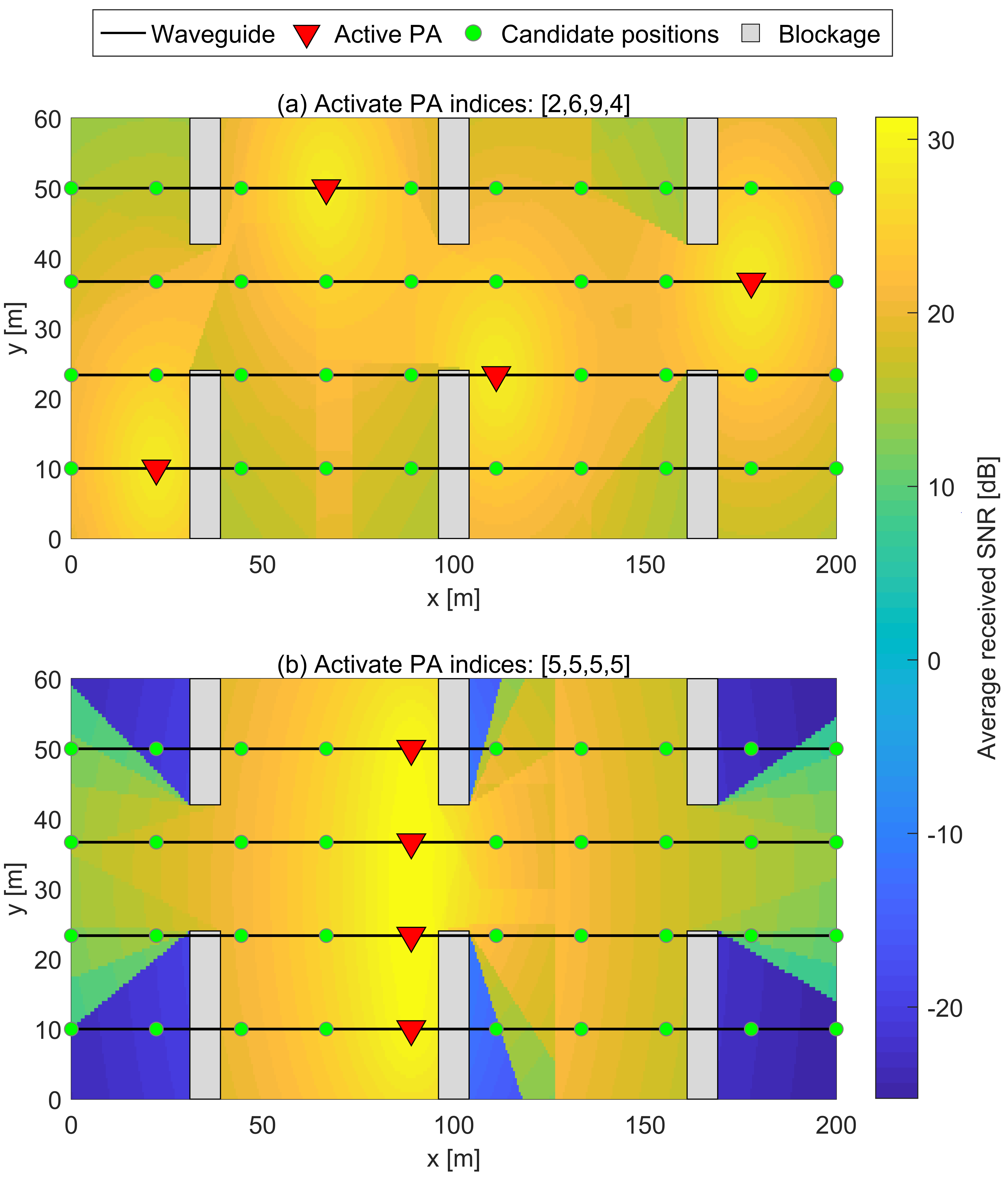}\\
        \captionsetup{justification=justified, singlelinecheck=false, font=small}	
        \caption{Top-view of the considered three-dimensional geometry-aware environment with average received SNR maps for two representative pinching-antenna deployments.} \label{fig:ga_rec_snr} \vspace{-5mm}
\end{figure} 

For a better illustration, we present the top-view of the considered geometry-aware scenario in Fig. \ref{fig:ga_rec_snr}. Fig. \ref{fig:ga_rec_snr} shows two examples of the resulting average received SNR maps under different pinching-antenna deployments, where Fig.\ref{fig:ga_rec_snr}(a) corresponds to the active pinching antenna indices $[2,6,9,4]$ (i.e., a more spatially distributed activation along the waveguides), while Fig. \ref{fig:ga_rec_snr}(b) corresponds to $[5,5,5,5]$ (i.e., a more aligned activation around a similar $x$-location). As seen, the distributed activation in Fig. \ref{fig:ga_rec_snr}(a) yields a more balanced SNR distribution across the region by providing diversity against distance loss and obstacle-induced LoS blockage. In contrast, the aligned activation in Fig. \ref{fig:ga_rec_snr}(b) results in pronounced low-SNR areas, especially near the edges and behind obstacles, since many user locations rely on blocked or long paths to the clustered active pinching antennas. This comparison further motivates geometry-aware deployment design, where exploiting spatial diversity is key to mitigating coverage ``dead zones''.

\begin{figure*}[!t]
	\begin{minipage}{0.32\linewidth}
		\centering
		\includegraphics[width=0.99\linewidth]{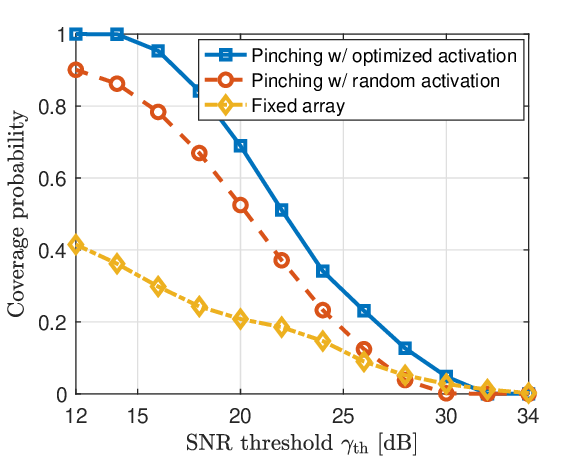}
		\captionsetup{justification=justified, singlelinecheck=false, font=small}	
		\caption{Average-SNR-threshold coverage probability versus the SNR threshold $\gamma_{\mathrm{th}}$ with $N = 4$ and $M = 10$.} \label{fig:ga_cov_max} \vspace{-3mm}
	\end{minipage}~~
	\begin{minipage}{0.32\linewidth} 
		\centering
		\includegraphics[width=0.99\linewidth]{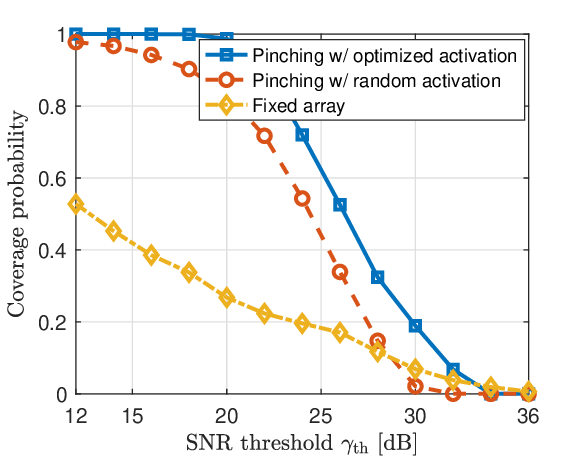}\\
		\captionsetup{justification=justified, singlelinecheck=false, font=small}	
        \caption{Average-SNR-threshold coverage probability versus the SNR threshold $\gamma_{\mathrm{th}}$ with $N = 8$ and $M = 20$.} \label{fig:ga_cov_max2}   \vspace{-3mm}
	\end{minipage}~~
	\begin{minipage}{0.32\linewidth}
			\centering
				\includegraphics[width=0.96\linewidth]{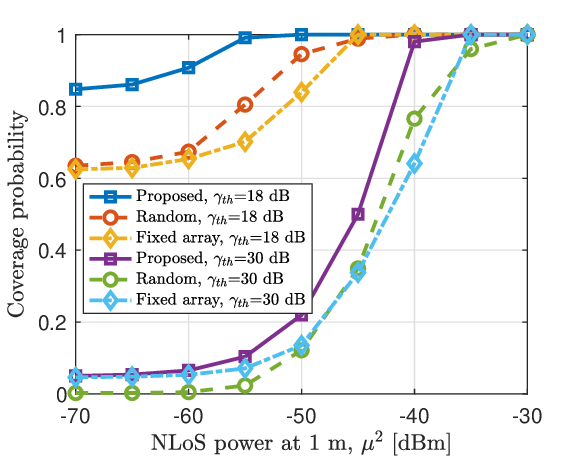}\\
        \captionsetup{justification=justified, singlelinecheck=false, font=small}	
        \caption{Sensitivity of the average-SNR-threshold coverage probability to the NLoS power parameter $\mu^2$.} \label{fig:ga_cov_max_mu}  \vspace{-3mm}
	\end{minipage} 
\end{figure*}

Fig. \ref{fig:ga_cov_max} illustrates the average-SNR-threshold coverage probability as a function of the SNR threshold $\gamma_{\mathrm{th}}$ under a uniform user distribution. For each $\gamma_{\mathrm{th}}$, the coverage probability is computed as the fraction of valid grid points whose average received SNR exceeds $\gamma_{\mathrm{th}}$, where grid points inside obstacle footprints are excluded. Three schemes are compared, namely, pinching antennas with optimized activation obtained by the proposed algorithm, pinching antennas with random activation where one candidate position per waveguide is selected uniformly at random, and a fixed $N$-element antenna array placed at the center of the communication region, where the antenna elements are half-wavelength spaced and are placed at the same height as the waveguides.

As observed in Fig.~\ref{fig:ga_cov_max}, the proposed optimized activation scheme consistently achieves the highest coverage probability over most of the considered threshold range. This indicates that, by adaptively selecting the active pinching locations across waveguides, the proposed method can better exploit spatial degrees of freedom to increase the fraction of locations meeting a target SNR requirement, especially in the presence of obstacle-induced LoS blocking and distance-dependent path loss. Random activation provides intermediate performance and generally outperforms the fixed-array baseline at moderate thresholds, since the pinching architecture still offers spatial diversity even without optimization. In contrast, the fixed array exhibits a noticeably lower coverage level at small and moderate $\gamma_{\mathrm{th}}$, reflecting the inherent limitation of a static, centrally located deployment in serving spatially distributed users under blockage. 

It is also seen that, when $\gamma_{\mathrm{th}}$ becomes very large, the fixed array may slightly outperform the pinching-based schemes. This behavior is expected because the fixed array concentrates all $N$ antenna elements within a very small aperture around the region center, thereby creating a localized ``hot spot'' with exceptionally high received SNR in its near vicinity. Such a strong local peak can contribute non-negligibly to the coverage probability under extremely stringent thresholds. In comparison, the pinching antennas are constrained to select one active radiating point on each waveguide from a discrete candidate set, and the resulting radiated power is spatially distributed across different waveguides and locations; while this improves overall spatial coverage at practical thresholds, it may yield a slightly smaller area that exceeds a very high $\gamma_{\mathrm{th}}$. As $\gamma_{\mathrm{th}}$ increases further, the coverage of all schemes eventually decreases toward zero, while the proposed scheme maintains a clear advantage over the two baselines over a wide operating range, highlighting the benefit of geometry-aware activation optimization for coverage enhancement.

Fig.~\ref{fig:ga_cov_max2} further evaluates the average-SNR-threshold coverage probability as a function of the SNR threshold $\gamma_{\mathrm{th}}$ for the case with $N=8$ waveguides/pinching antennas and $M=20$ candidate positions per waveguide. Compared with Fig.~\ref{fig:ga_cov_max}, which corresponds to $(N,M)=(4,10)$, increasing the spatial degrees of freedom to $(8,20)$ yields a clear upward shift of all coverage curves, with the improvement being most pronounced in the low-to-moderate SNR regime. This trend is expected because activating more pinching antennas provides stronger power aggregation and additional spatial diversity, while a denser candidate set offers finer location resolution, enabling the optimized activation to select radiating points that better align with favorable LoS corridors and mitigate obstacle-induced shadowing. Consequently, for moderate $\gamma_{\mathrm{th}}$, a larger portion of the region can satisfy the SNR requirement, and the performance advantage of optimized activation becomes more evident. In contrast, when $\gamma_{\mathrm{th}}$ is very large, the coverage probability of all schemes remains small and becomes less sensitive to further increasing $(N,M)$. This is because the high-threshold regime is dominated by a few localized high-SNR ``hot spots'', whose existence and spatial extent are mainly dictated by the strongest available LoS links and the minimum achievable path loss in the given geometry. Once such hot spots are formed, adding more antennas or candidate points has limited impact on expanding the area that exceeds an extremely stringent $\gamma_{\mathrm{th}}$. As a result, although enlarging $(N,M)$ yields only marginal improvement at high SNR thresholds, it significantly enhances the coverage in the practical SNR regime.


Fig. \ref{fig:ga_cov_max_mu} evaluates the impact of the NLoS power level $\mu^2$ on the average-SNR-threshold coverage probability for the geometry-aware scenario. Two representative SNR thresholds, $\gamma_{\mathrm{th}}=18$~dB and $\gamma_{\mathrm{th}}=30$~dB, are considered. Three schemes are compared, including pinching antennas with optimized activation obtained by the proposed algorithm, pinching antennas with random activation, and a fixed centered antenna-array baseline.
As shown in Fig. \ref{fig:ga_cov_max_mu}, increasing $\mu^2$ monotonically improves the coverage probability for all schemes, since a stronger NLoS component lifts the average received SNR even when LoS links are blocked. At the moderate threshold $\gamma_{\mathrm{th}}=18$~dB, the coverage is already high for the proposed scheme over a wide range of $\mu^2$, while the fixed-array and random-activation baselines exhibit a more visible dependence on $\mu^2$, indicating that the proposed optimized activation is less reliant on NLoS strength due to its ability to select favorable activation locations. In contrast, at the stringent threshold $\gamma_{\mathrm{th}}=30$~dB, the coverage is more sensitive to $\mu^2$ and remains low when $\mu^2$ is weak, because meeting such a high SNR requirement typically requires either strong LoS links or a sufficiently large NLoS power. As $\mu^2$ increases, the performance gap between the proposed scheme and the two baselines is most pronounced in the intermediate regime, where NLoS is neither negligible nor dominant; in this regime, optimized activation can still exploit geometry to form higher-SNR regions, whereas random activation and the fixed array are less effective. When $\mu^2$ becomes sufficiently large, the curves approach one and the differences diminish, since the SNR becomes increasingly dominated by the NLoS component and the benefit of geometry-aware activation is reduced.

\begin{table}[!t]
\centering
\caption{Comparison with the globally optimal solution for the
average-SNR-threshold coverage maximization problem.}
\renewcommand{\arraystretch}{1.2}
\label{tab:coverage_global}
\begin{tabular}{c|cc|c|cc}
\hline
\multirow{2}{*}{$N$}
& \multicolumn{2}{c|}{Coverage probability}
& \multirow{2}{*}{Gap [p.p.]}
& \multicolumn{2}{c}{CPU time [s]} \\
& Proposed & BnB & & Proposed & BnB \\
\hline
3 & 0.34789 & 0.34789 & 0     & 0.0138 & 22.16  \\
4 & 0.59155 & 0.59155 & 0     & 0.0105 & 55.67  \\
5 & 0.77183 & 0.77295 & 0.112 & 0.0041 & 234.91 \\
6 & 0.94507 & 0.94743 & 0.236 & 0.0061 & 428.63 \\
\hline
\end{tabular}
\end{table}

To further evaluate the solution quality and computational efficiency of
the proposed coordinate-ascent method, we compare it with the globally
optimal branch-and-bound (BnB) solution of the MILP in \eqref{prob:cov_milp}. For the proposed low-complexity method, $10$ feasible initial activation
matrices are used and the best converged solution is retained.
We consider
a representative small-scale setting with $D_x=80$~m, $D_y=40$~m,
$N_h\times N_v=40\times20$, $M=6$, and
$\gamma_{\rm th}=30$~dB. The comparison is summarized in
Table~\ref{tab:coverage_global}. 
As shown in Table~\ref{tab:coverage_global}, the proposed
coordinate-ascent method achieves the same coverage probability as the
globally optimal solution for $N=3$, and $4$. For $N=6$, the
coverage-probability gap is only about $0.236$ percentage points.
In contrast, the computational cost differs substantially. The proposed
method requires less than $0.014$~s in all considered cases, whereas the
globally optimal BnB solver requires from $22.16$~s to $428.63$~s.
These results show that the proposed method can achieve near-global
coverage performance in the considered setting with substantially lower
computational complexity.

\begin{figure}[!t]
	\centering
	\includegraphics[width=0.76\linewidth]{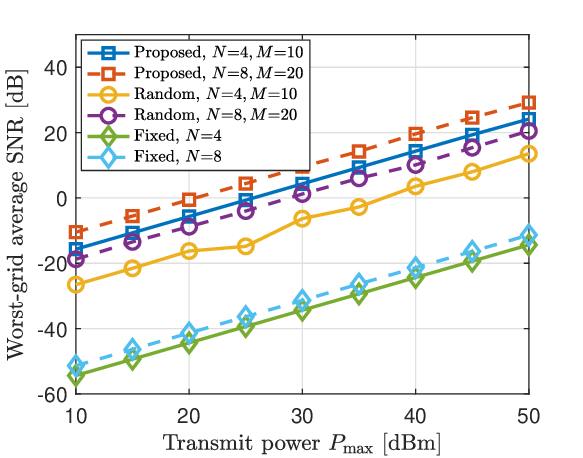}\\
        \captionsetup{justification=justified, singlelinecheck=false, font=small}	
        \caption{Worst-grid average-SNR performance versus transmit power $P$.} \label{fig:ga_worst_Pmax} \vspace{-5mm}
\end{figure} 

Fig.~\ref{fig:ga_worst_Pmax} plots the achieved worst-grid average SNR as a function of the transmit power $P$ for the geometry-aware max-min average-SNR design. Here, the worst-grid average SNR is defined as the minimum, over all valid grids (excluding those inside the obstacle footprints), of the average received SNR. We compare three schemes, namely the proposed pinching-antenna design with optimized activation obtained by the proposed algorithm, pinching antennas with random activation, and a fixed centered $N$-element antenna array with half-wavelength spacing deployed at the same height as the waveguides. To illustrate the impact of deployment freedom, the results are shown for two configurations, i.e., $(N,M)=(4,10)$ and $(N,M)=(8,20)$, where $N$ denotes the number of waveguides/active pinching antennas and $M$ denotes the number of candidate locations per waveguide.

As observed in Fig.~\ref{fig:ga_worst_Pmax}, all schemes increase approximately linearly (in dB) with $P$, as the average received SNR scales proportionally with transmit power under the adopted average-SNR model. More importantly, for both $(N,M)=(4,10)$ and $(N,M)=(8,20)$, the proposed optimized activation consistently achieves the largest worst-grid SNR across the entire power range, demonstrating that geometry-aware activation can effectively mitigate the most unfavorable locations caused by distance-dependent path loss and obstacle-induced LoS blocking. Random activation provides a noticeable improvement over the fixed-array baseline but remains clearly below the optimized design, indicating that while the pinching architecture offers inherent spatial diversity, explicit optimization is crucial for substantially lifting the worst-case performance. The fixed centered array yields the lowest worst-grid SNR, reflecting the limitation of a static deployment in a blockage-rich environment: the max-min objective is dominated by a small set of severely shadowed or farthest grids, and without the ability to relocate radiating points, the fixed array cannot eliminate these geometric ``dead zones,'' leading to a persistently low minimum SNR even when $P$ increases.

Comparing the two configurations in Fig.~\ref{fig:ga_worst_Pmax}, increasing the spatial degrees of freedom from $(N,M)=(4,10)$ to $(N,M)=(8,20)$ yields a clear improvement for the pinching-antenna schemes, where the curves of both optimized activation and random activation are shifted upward by roughly $\sim 10$ dB over the considered power range. This gain is expected since more waveguides/active points provide stronger power aggregation and richer spatial diversity, while a denser candidate set enables finer-grained placement so that the activation can better target and alleviate the most unfavorable (worst-grid) locations created by distance-dependent attenuation and obstacle-induced shadowing. In contrast, the fixed-array baseline benefits much less from increasing $N$, exhibiting only a limited improvement. The reason is that the worst-grid metric is dominated by a few bottleneck grids that lie in deep shadow or far-field corners, and a centrally deployed fixed array cannot relocate its radiation points to mitigate these geometric dead zones. As a result, additional antennas mainly enhance already-strong directions but do not substantially raise the minimum SNR across the entire region. Overall, these results confirm that additional deployment freedom translates more effectively into worst-case performance gains for the pinching architecture, especially when combined with geometry-aware activation optimization.

\begin{figure}[!t]
	\centering
	\includegraphics[width=0.76\linewidth]{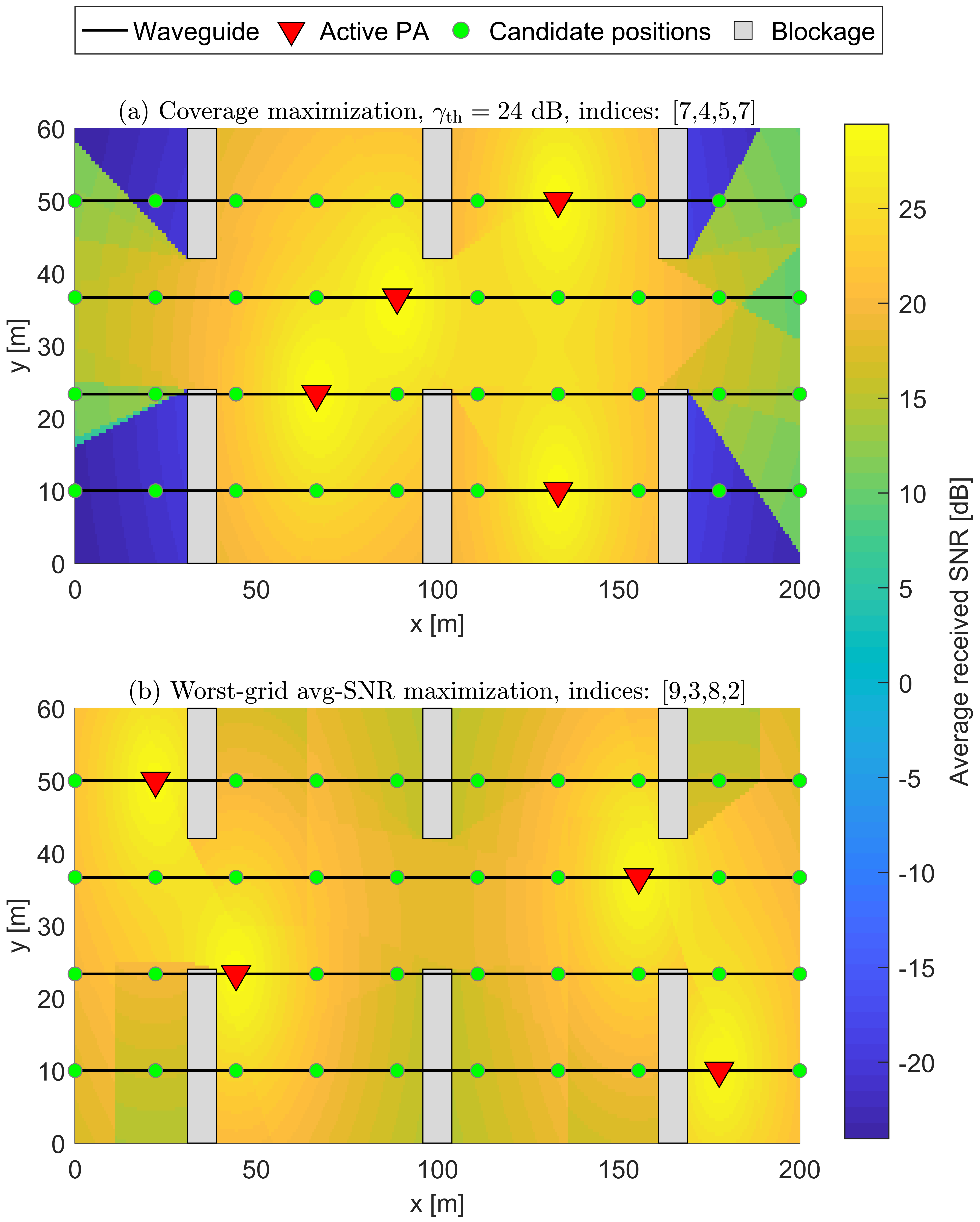}\\
        \captionsetup{justification=justified, singlelinecheck=false, font=small}	
        \caption{Top-view of the considered geometry-aware environment with optimized pinching-antenna deployments for two representative design objectives: (a) average-SNR-threshold coverage maximization with $\gamma_{\mathrm{th}}=24$ dB; (b) worst-grid average-SNR maximization.} \label{fig:ga_two_obj_maps} \vspace{-5mm}
\end{figure} 

Fig.~\ref{fig:ga_two_obj_maps} visualizes the optimized pinching-antenna deployments obtained for the two geometry-aware design problems in the same blockage environment. The background colormap shows the resulting spatial map of the average received SNR evaluated on a uniform grid after excluding the grid points that fall inside the blockage footprints. A single shared colorbar is used for both subfigures to facilitate a direct visual comparison of the achieved SNR levels and the spatial patterns.

In Fig.~\ref{fig:ga_two_obj_maps}(a), the deployment is optimized for the average-SNR-threshold coverage maximization with a target threshold $\gamma_{\mathrm{th}}=24$ dB, and the corresponding activated pinching antenna indices are $[7,4,5,7]$. The selected pinching locations tend to create relatively strong SNR regions that cover a large portion of the feasible area, while accepting that some peripheral regions and obstacle-shadowed corners remain below the threshold. In contrast, Fig.~\ref{fig:ga_two_obj_maps}(b) shows the deployment optimized for the worst-grid average-SNR maximization, with activated indices $[9,3,8,2]$. Compared with (a), the activated points in (b) are more ``spread out'' across the waveguides and tend to shift toward locations that improve the weakest regions, leading to a more balanced SNR distribution with fewer pronounced low-SNR areas.

The different activation patterns in Fig.~\ref{fig:ga_two_obj_maps} directly reflect the distinct design philosophies of the two objectives. Coverage maximization is threshold-driven: it aims to maximize the fraction of grids whose SNR exceeds $\gamma_{\mathrm{th}}$, and therefore it is beneficial to position some radiating points so as to enlarge high-SNR regions around favorable LoS corridors, even if this leaves a few highly shadowed or far-away grids unimproved. This behavior is visible in Fig.~\ref{fig:ga_two_obj_maps}(a), where strong regions dominate much of the interior while the corners behind blockages still exhibit low SNR. On the other hand, worst-grid maximization is fairness-oriented: it is governed by the minimum SNR over all grids, and thus it forces the activation to explicitly mitigate ``dead zones'' by improving the most disadvantaged locations. Consequently, the optimized deployment in Fig.~\ref{fig:ga_two_obj_maps}(b) tends to avoid concentrating radiating points on already-strong areas and instead reallocates degrees of freedom to lift the weakest parts of the region. Overall, Fig.~\ref{fig:ga_two_obj_maps} highlights that, even under the same hardware constraints and blockage geometry, geometry-aware activation optimization can yield markedly different pinching antenna placements depending on whether the design goal prioritizes coverage at a target SNR level or worst-case spatial fairness.

\begin{figure}[!t]
	\centering
	\includegraphics[width=0.76\linewidth]{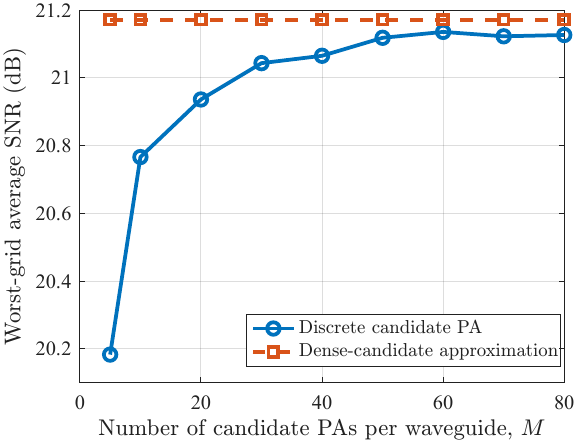}\\
        \captionsetup{justification=justified, singlelinecheck=false, font=small}	
        \caption{Worst-grid average SNR versus the number of candidate pinching antennas per waveguide, where $M=200$ is used as a numerical approximation to continuously adjustable radiation locations.} \label{fig:candidate_density}
\end{figure} 

We next investigate the impact of the finite candidate set on the
geometry-aware design, as studied in \cite{tyrovolas2026many}. The candidate pinching antennas are uniformly
distributed along each waveguide, such that increasing the number of
candidates $M$ simultaneously reduces the candidate spacing. Specifically,
the spacing is given by $\Delta_x=D_x/(M+1)$. We consider the worst-grid
average-SNR maximization problem as a representative example and use a
dense deployment with $M=200$ as a numerical approximation to
continuously adjustable radiation locations. 
As shown in Fig.~\ref{fig:candidate_density}, the worst-grid average SNR
generally improves as $M$ increases and gradually approaches the
dense-candidate benchmark. In particular, the performance improves
significantly when increasing $M$ from a sparse deployment, whereas the
additional gain becomes relatively small once the candidate set is
sufficiently dense. This demonstrates that the performance loss caused by
the discrete activation architecture can be effectively reduced by
decreasing the candidate spacing, while an excessively dense deployment
provides only marginal additional benefit.

To further evaluate the solution quality of the proposed
low-complexity method, Table~\ref{tab:global_comparison}
compares it with the globally optimal BnB 
solution of the MILP in \eqref{prob:wg_mip} in small-scale settings with $D_x=80$~m, $D_y=40$~m, $N_h\times N_v=60\times30$, and $M=10$.  
For the proposed low-complexity method, $10$ feasible initial activation
matrices are used and the best converged solution is retained.
As shown, the proposed method achieves almost
the same worst-grid average SNR as the global solution. In
particular, the average optimality gap is no more than
$0.092$ dB for $N=3$--$6$. In contrast, the computational
difference becomes increasingly pronounced as $N$ grows.
The runtime of the proposed method remains below $0.1$ s,
whereas the average BnB runtime increases from $3.44$ s for
$N=3$ to $30.83$ s for $N=6$. These results demonstrate that
the proposed method provides a favorable tradeoff between
solution quality and computational complexity, despite not
having a global-optimality guarantee.

\begin{table}[!t]
\centering
\caption{Comparison with the globally optimal solution for the
worst-grid average-SNR maximization problem.}
\renewcommand{\arraystretch}{1.2}
\label{tab:global_comparison}
\begin{tabular}{c|cc|c|cc}
\hline
\multirow{2}{*}{$N$}
& \multicolumn{2}{c|}{Worst-grid average SNR [dB]}
& \multirow{2}{*}{Gap [dB]}
& \multicolumn{2}{c}{CPU time [s]} \\
& Proposed & BnB & & Proposed &  BnB \\
\hline
3 & 24.237 & 24.238 & 0.001 & 0.060 & 3.444 \\
4 & 26.376 & 26.376 & 0.001 & 0.069 & 9.828 \\
5 & 27.805 & 27.867 & 0.061 & 0.088 & 19.248 \\
6 & 28.917 & 29.009 & 0.092 & 0.092 & 30.834 \\
\hline
\end{tabular}
\end{table}

We further evaluate the impact of in-waveguide attenuation on the proposed
geometry-aware designs. As discussed in Part I and also investigated in
\cite{tyrovolas2025performance}, signal attenuation during propagation
along the waveguide can affect the effective radiated power of pinching
antennas located at different distances from the feed point. Following the
attenuation model adopted in Part I, the corresponding attenuation factor
is incorporated into the precomputed candidate-specific SNR maps. We
compare an attenuation-aware design, which accounts for the waveguide loss
during activation optimization, with an attenuation-ignorant design, which
optimizes the activation pattern without considering attenuation and is
then evaluated under the same attenuated channel. The results are averaged
over 10 independent blockage realizations.

\begin{table}[!t]
\centering 
\caption{Impact of in-waveguide attenuation on the proposed
geometry-aware designs.}
\renewcommand{\arraystretch}{1.15}
\label{tab:attenuation_sensitivity}
\begin{tabular}{c|ccc|ccc}
\hline
$\alpha_{\rm g}$
& \multicolumn{3}{c|}{Maximum coverage probability}
& \multicolumn{3}{c}{Worst-grid average SNR [dB]}\\
dB/m & Aware & Ignorant & Gap [p.p.]
& Aware & Ignorant & Gap\\
\hline
0    & 1.000 & 1.000 & 0     & 28.313 & 28.313 & 0     \\
0.01 & 1.000 & 0.999 & 0.016 & 27.998 & 27.924 & 0.074 \\
0.02 & 1.000 & 0.999 & 0.085 & 27.647 & 27.475 & 0.172 \\
0.04 & 1.000 & 0.995 & 0.415 & 26.883 & 26.540 & 0.342 \\
0.06 & 1.000 & 0.986 & 1.355 & 26.135 & 25.610 & 0.525 \\
0.08 & 1.000 & 0.967 & 3.247 & 25.413 & 24.702 & 0.711 \\
0.10 & 1.000 & 0.933 & 6.689 & 24.570 & 23.821 & 0.749 \\
\hline
\end{tabular}
\end{table}

As shown in Table~\ref{tab:attenuation_sensitivity}, the impact of
neglecting in-waveguide attenuation is small when the waveguide loss is
low. For example, at $\alpha_{\rm g}=0.01$~dB/m, the coverage and
worst-grid average-SNR gaps are only $0.016$ percentage points and
$0.074$~dB, respectively. As the attenuation increases, however, the
performance loss of the attenuation-ignorant design becomes more
pronounced. At $\alpha_{\rm g}=0.10$~dB/m, the corresponding gaps
increase to $6.689$ percentage points and $0.749$~dB. These results
indicate that the attenuation-free model is a reasonable approximation
for low-loss waveguides, whereas attenuation should be
explicitly accounted for in longer or more lossy deployments.
We note that, for long waveguides, the accumulated attenuation can also be alleviated
through a segmented-waveguide architecture \cite{ouyang2026swan}, where a long waveguide is
divided into multiple shorter segments with separate feed points, thereby
reducing the maximum in-waveguide propagation distance before radiation.

\vspace{-0mm}
\section{Conclusion} \label{sec: conclusion}
This paper investigated geometry-aware network-level design for pinching-antenna systems in blockage-rich environments. By discretizing the service region into grids and modeling the environment with three-dimensional cuboid obstacles, we introduced a deterministic LoS visibility indicator to capture geometry-induced LoS/NLoS transitions between candidate pinching-antenna locations and grid points. Under a practical discrete activation model, we derived tractable per-grid average-SNR expressions whose geometry-dependent components can be precomputed offline, enabling an efficient offline–online optimization workflow.
Building on these metrics, we formulated two geometry-aware network-level optimization problems: average-SNR-threshold coverage maximization and fairness-oriented worst-grid average-SNR maximization. Efficient solution methods were developed, including a low-complexity coordinate-ascent algorithm for the coverage maximization, and an epigraph-based low-complexity bisection method with coordinate-wise feasibility search for the worst-grid average SNR maximization. Numerical results demonstrated that geometry-aware pinching-antenna activation can improve both coverage probability and worst-grid SNR over representative fixed-array and random-activation baselines, while also revealing how the coverage–robustness tradeoff evolves with blockage layouts and key system parameters.
Future work may extend the proposed framework to multi-cell deployments, and imperfect or dynamically learned environment maps, as well as to hybrid traffic-and-geometry-aware objectives that better reflect time-varying network operations.

\vspace{-0mm}
\begin{appendices}
    \section{Proof of Lemma \ref{lem: per-grid snr}} \label{appd: per-grid snr} 
    For grid $(u,v)\in\Omega$,
    let $\Gamma(\boldsymbol{\psi}_{u,v};\mathbf{A})$ denote the instantaneous SNR
    under activation $\mathbf{A}$ as defined in \eqref{eq:snr_mrt}. The corresponding per-grid
    average received SNR is defined as
    \begin{equation}
        \bar{\Gamma}(\boldsymbol{\psi}_{u,v};\mathbf{A})
        \triangleq
        \mathbb{E}\!\left[\Gamma(\boldsymbol{\psi}_{u,v};\mathbf{A})\right]
        =
        \rho\,\mathbb{E}\!\left[
            \big\|\tilde{\boldsymbol{h}}(\boldsymbol{\psi}_{u,v};\mathbf{A})\big\|^2
        \right],
        \label{eq:avg_snr_def}
    \end{equation}
    where $\rho\triangleq P/\sigma^2$ and the expectation is taken over the
    small-scale fading.
    
    Using \eqref{eq:heff_def}--\eqref{eq:heff_vec} and $\sum_{m=1}^{M}a_{n,m}=1$,
    we obtain
    \begin{equation}
        \mathbb{E}\!\left[
            \big\|\tilde{\boldsymbol{h}}(\boldsymbol{\psi}_{u,v};\mathbf{A})\big\|^2
        \right]
        =
        \sum_{n=1}^{N}\sum_{m=1}^{M} a_{n,m}\,
        \bar{G}_{n,m}(\boldsymbol{\psi}_{u,v}),
        \label{eq:avg_hnorm}
    \end{equation}
    where
    \begin{equation}
        \bar{G}_{n,m}(\boldsymbol{\psi}_{u,v})
        \triangleq
        \mathbb{E}\!\left[\big|h_{n,m}(\boldsymbol{\psi}_{u,v})\big|^2\right]
    \end{equation}
    denotes the average channel power from candidate $(n,m)$ to grid $(u,v)$.
    Recalling \eqref{eq:hn-decomposition-geo} and noting that
    $h_{n,m}^{\mathrm{NLoS}}(\boldsymbol{\psi}_{u,v})$ is zero-mean and independent
    of the deterministic LoS term, we have
    \begin{align}
        \bar G_{n,m}(\boldsymbol{\psi}_{u,v}) &=
        \chi_{n,m}(\boldsymbol{\psi}_{u,v})\,
        \big|h_{n,m}^{\mathrm{LoS}}(\boldsymbol{\psi}_{u,v})\big|^2 \notag\\
        &\quad + \mathbb{E} \left[\big|h_{n,m}^{\mathrm{NLoS}}(\boldsymbol{\psi}_{u,v})\big|^2\right].
        \label{eq:Gbar_decomp}
    \end{align}
    From \eqref{eq:hn-los-geo}, we have
    \begin{equation}
        \big|h_{n,m}^{\mathrm{LoS}}(\boldsymbol{\psi}_{u,v})\big|^2
        =
        \frac{\eta}{r_{n,m}^2(\boldsymbol{\psi}_{u,v})}.
    \end{equation}
    Moreover, with \eqref{eq:gnl-grid}, the average NLoS power is given by
    \begin{equation}
        \mathbb{E}\!\left[\big|h_{n,m}^{\mathrm{NLoS}}(\boldsymbol{\psi}_{u,v})\big|^2\right]
        =
        \sum_{\ell=1}^{N_c}\frac{\mu_\ell^2}{r_{n,m}^2(\boldsymbol{\psi}_{u,v})}
        =
        \frac{\mu^2}{r_{n,m}^2(\boldsymbol{\psi}_{u,v})},
    \end{equation}
    where $\mu^2\triangleq\sum_{\ell=1}^{N_c}\mu_\ell^2$ denotes the aggregated NLoS channel power.
    Therefore, we have
    \begin{equation}
        \bar G_{n,m}(\boldsymbol{\psi}_{u,v})
        =
        \frac{\chi_{n,m}(\boldsymbol{\psi}_{u,v})\,\eta+\mu^2}
        {r_{n,m}^2(\boldsymbol{\psi}_{u,v})}.
        \label{eq:Gbar_piecewise}
    \end{equation}
    Substituting \eqref{eq:avg_hnorm} and \eqref{eq:Gbar_piecewise} into \eqref{eq:avg_snr_def} 
    yields \eqref{eq:avg_snr_final}. This completes the proof.
    \hfill $\blacksquare$

    \section{Proof of Proposition \ref{prop:NP_hard_MILP_coverage}} \label{appd:NP_hard_MILP_coverage}
    We prove NP-hardness by a polynomial-time reduction from the
    \emph{Maximum Coverage} (MC) problem, which is NP-hard \cite{khuller1999budgeted}.
    To this end, let us first present the definition of the MC problem.
    Given a universe $\mathcal{U}=\{1,\dots,G\}$ and a collection of subsets
    $\{\mathcal{S}_j\}_{j=1}^{J}$ with $\mathcal{S}_j\subseteq \mathcal{U}$,
    together with an integer budget $k$, the MC problem is given by\footnote{The standard MC problem uses the constraint $|\mathcal{J}|\le k$. For technical convenience, we consider the exact-$k$ variant, i.e., $|\mathcal{J}|=k$, which is also NP-hard. This choice will simplify the subsequent reduction by aligning with the ``one active pinching antenna per waveguide'' structure in our setting.}
    \begin{equation} \label{eq:MC_def_exact}
        \max_{\mathcal{J}\subseteq \{1,\dots,J\},\,|\mathcal{J}|= k}
        \ \bigg|\bigcup_{j\in\mathcal{J}} \mathcal{S}_j\bigg|.
    \end{equation}
    As seen, the MC problem aims to select a subcollection $\mathcal{J}$ with $|\mathcal{J}|\le k$ so that
    the union $\bigcup_{j\in\mathcal{J}}\mathcal{S}_j$ covers as many distinct elements in $\mathcal{U}$ as possible. 
    
    Recall that to prove NP-hardness of \eqref{prob:cov_milp}, it suffices to show that a known NP-hard problem
    reduces to \eqref{prob:cov_milp} in polynomial time, i.e., $\mathrm{MC}\le_p \eqref{prob:cov_milp}$ \cite{arora2009computational}. Here, we write $A \le_p B$ if there exists a polynomial-time computable mapping $f$ from instances of $A$ to instances of $B$ such that the answer to an instance $x$ of $A$ can be obtained from the answer to $f(x)$ of $B$ in polynomial time.
    Next, we show how to map any MC instance $(\mathcal{U},\{\mathcal{S}_j\}_{j=1}^{J},k)$ to a special instance of
    \eqref{prob:cov_milp} in polynomial time. 
    Given any MC instance $(\mathcal{U},\{\mathcal{S}_j\}_{j=1}^{J},k)$, we construct a special instance of
    \eqref{prob:cov_milp} such that (i)  each waveguide activates exactly one candidate (as in our model),
    yet (ii) the number of distinct candidates effectively selected is at most $k$, matching the MC budget.
    Moreover, a grid is ``covered'' in MC if and only if its average SNR reaches the threshold in \eqref{prob:cov_milp}.
    The construction is as follows.
    \begin{itemize}
    \item \emph{Grids $\leftrightarrow$ MC elements:}
    Let the grid index set $\Omega$ contain $|\Omega|=G$ grids, in one-to-one correspondence with the
    elements in $\mathcal{U}$. Throughout, ``grid $i$'' represents element $i\in\mathcal{U}$.
    
    \item \emph{Budget $\leftrightarrow$ number of waveguides:}
    Set the number of waveguides to $N=k$.
    
    \item \emph{Candidates $\leftrightarrow$ MC sets:}
    On each waveguide $n\in\{1,\dots,N\}$, set the number of candidate positions to $M=J$.
    Candidate $m\in\{1,\dots,J\}$ corresponds to the MC set $\mathcal{S}_m$.

   \item \emph{Encoding set membership via SNR contributions:}
    Let $\gamma_{\mathrm{th}}>0$ be the SNR threshold in \eqref{prob:cov_milp}.
    Define the (precomputable) per-grid average-SNR contribution of candidate $(n,m)$ to grid $i\in\Omega$ as\footnote{We note that $\overline{\Gamma}_{n,m}(i)$ is \emph{artificially defined} in the reduction (not derived from a physical channel model) to encode the set-membership relation $i\in\mathcal{S}_m$ via the threshold $\gamma_{\mathrm{th}}$.
}
    \begin{equation}
    \label{eq:NP_coeff_def}
    \overline{\Gamma}_{n,m}(i)
    \triangleq
    \begin{cases}
    \gamma_{\mathrm{th}}, & \text{if } i\in \mathcal{S}_m,\ m\in\{1,\dots,J\},\\
    0, & \text{otherwise},
    \end{cases}
    \end{equation}
    so that activating candidate $m\in\{1,\dots,J\}$ contributes exactly $\gamma_{\mathrm{th}}$ to precisely those grids
    whose indices belong to $\mathcal{S}_m$, and contributes $0$ otherwise.
    Then, under an activation matrix $\mathbf{A}=[a_{n,m}]$, define the per-grid average SNR as
    \begin{equation}
    \label{eq:NP_linear_SNR}
    \overline{\Gamma}(i;\mathbf{A})
    \triangleq
    \sum_{n=1}^{N}\sum_{m=1}^{M} a_{n,m}\,\overline{\Gamma}_{n,m}(i),
    \quad \forall i\in\Omega.
    \end{equation}
    \end{itemize}
    The above construction is computable in time polynomial in $(G,J,k)$.
    Then, with \eqref{eq:NP_linear_SNR}, the MILP \eqref{prob:cov_milp} reduces to
    \begin{subequations}
    \label{eq:NP_MILP_instance}
    \begin{align}
    \max_{\mathbf{A},\{c_i\}} \ & \sum_{i\in\Omega} c_i \label{eq:NP_MILP_instance_a}\\
    \text{s.t.}\quad
    & \overline{\Gamma}(i;\mathbf{A}) \ge \gamma_{\mathrm{th}}\,c_i,\ \forall i\in\Omega, \label{eq:NP_MILP_instance_b}\\
    & \sum_{m=1}^{M} a_{n,m}=1,\ \forall n=1,\dots,N, \label{eq:NP_MILP_instance_c}\\
    & a_{n,m}\in\{0,1\},\ \forall n,m, \label{eq:NP_MILP_instance_d1}\\
    & c_i\in\{0,1\},\ \forall i\in\Omega. \label{eq:NP_MILP_instance_d2}
    \end{align}
    \end{subequations}
     
    Let $\mathrm{OPT}_{\mathrm{MC}}$ and $\mathrm{OPT}_{\mathrm{MILP}}$ denote the optimal objective values of
    \eqref{eq:MC_def_exact} and \eqref{eq:NP_MILP_instance}, respectively. We show that
    $\mathrm{OPT}_{\mathrm{MILP}}=\mathrm{OPT}_{\mathrm{MC}}$.
    
    \smallskip
    \noindent\emph{(i) Any MC solution induces a feasible MILP solution with the same objective.}
    Take any $\mathcal{J}\subseteq\{1,\dots,J\}$ with $|\mathcal{J}|=k$.
    Construct $\mathbf{A}$ by assigning the $k$ selected set indices in $\mathcal{J}$ to the $N=k$ waveguides,
    i.e., for each $n\in\{1,\dots,k\}$ choose one distinct $j_n\in\mathcal{J}$ and set $a_{n,j_n}=1$ (and $a_{n,m}=0$ for $m\neq j_n$),
    so that \eqref{eq:NP_MILP_instance_c}--\eqref{eq:NP_MILP_instance_d1} hold. 
    Define
    \begin{equation}
    \label{eq:c_choice}
    c_i \triangleq \mathbbm{1}\!\bigg\{ i\in \bigcup_{j\in\mathcal{J}} \mathcal{S}_j \bigg\},\quad \forall i\in\Omega.
    \end{equation}
    If $c_i=1$, then there exists $j\in\mathcal{J}$ such that $i\in\mathcal{S}_j$.
    By \eqref{eq:NP_coeff_def} and \eqref{eq:NP_linear_SNR}, at least one selected term equals $\gamma_{\mathrm{th}}$,
    hence $\overline{\Gamma}(i;\mathbf{A})\ge \gamma_{\mathrm{th}}$ and \eqref{eq:NP_MILP_instance_b} holds.
    If $c_i=0$, then \eqref{eq:NP_MILP_instance_b} holds trivially since $\overline{\Gamma}(i;\mathbf{A})\ge 0$.
    Therefore, $(\mathbf{A},\{c_i\})$ is feasible for \eqref{eq:NP_MILP_instance}, and its objective value is
    \begin{align}
    \sum_{i\in\Omega} c_i = \bigg|\bigcup_{j\in\mathcal{J}} \mathcal{S}_j\bigg|.
    \end{align}
    Maximizing over all feasible $\mathcal{J}$ yields $\mathrm{OPT}_{\mathrm{MILP}}\ge \mathrm{OPT}_{\mathrm{MC}}$.
    
    \smallskip
    \noindent\emph{(ii) Any feasible MILP solution induces an MC solution with no smaller objective.}
    Conversely, take any feasible $(\mathbf{A},\{c_i\})$ of \eqref{eq:NP_MILP_instance}.
    By \eqref{eq:NP_MILP_instance_c}--\eqref{eq:NP_MILP_instance_d1}, each waveguide selects exactly one candidate.
    Let $\mathcal{J}(\mathbf{A})\subseteq\{1,\dots,J\}$ denote the set of indices selected by at least one waveguide under $\mathbf{A}$.
    Since there are $N=k$ waveguides, we have $|\mathcal{J}(\mathbf{A})|\le k$.
    If $|\mathcal{J}(\mathbf{A})|<k$, we can add arbitrary extra indices from $\{1,\dots,J\}\setminus \mathcal{J}(\mathbf{A})$
    to obtain a set $\widehat{\mathcal{J}}(\mathbf{A})$ with $|\widehat{\mathcal{J}}(\mathbf{A})|=k$.
    Clearly,
    $\bigcup_{j\in\widehat{\mathcal{J}}(\mathbf{A})}\mathcal{S}_j \supseteq \bigcup_{j\in\mathcal{J}(\mathbf{A})}\mathcal{S}_j$,
    so $\widehat{\mathcal{J}}(\mathbf{A})$ is feasible for \eqref{eq:MC_def_exact} and can only increase (or keep) the union size.
    
    Now, consider any grid $i$ with $c_i=1$. Constraint \eqref{eq:NP_MILP_instance_b} implies
    $\overline{\Gamma}(i;\mathbf{A})\ge \gamma_{\mathrm{th}}$.
    By \eqref{eq:NP_linear_SNR}--\eqref{eq:NP_coeff_def}, this is possible only if there exists at least one selected candidate
    $m\in\mathcal{J}(\mathbf{A})$ such that $i\in\mathcal{S}_m$.
    Hence,
    \begin{equation}
    \label{eq:covered_inclusion}
    \{i\in\Omega:\ c_i=1\}
    \subseteq
    \bigcup_{m\in\mathcal{J}(\mathbf{A})}\mathcal{S}_m
    \subseteq
    \bigcup_{j\in\widehat{\mathcal{J}}(\mathbf{A})}\mathcal{S}_j,
    \end{equation}
    which yields
    \begin{align}
    \sum_{i\in\Omega} c_i \le \bigg|\bigcup_{j\in\widehat{\mathcal{J}}(\mathbf{A})}\mathcal{S}_j\bigg|.
    \end{align}
    Therefore, the MC solution $\widehat{\mathcal{J}}(\mathbf{A})$ achieves an objective value no smaller than the MILP objective.
    Taking the maximum over all MILP feasible solutions gives $\mathrm{OPT}_{\mathrm{MC}}\ge \mathrm{OPT}_{\mathrm{MILP}}$.
    
    Finally, combining (i) and (ii), we obtain $\mathrm{OPT}_{\mathrm{MILP}}=\mathrm{OPT}_{\mathrm{MC}}$ for the constructed instance, and the construction and mappings are polynomial-time. Therefore, if \eqref{prob:cov_milp} could be solved in polynomial time, then MC could also be solved in polynomial time. Since MC is NP-hard, \eqref{prob:cov_milp} is NP-hard
    (even restricted to instances of the form \eqref{eq:NP_MILP_instance}), and hence \eqref{prob:cov_milp} is NP-hard in general.
    This completes the proof.
    \hfill $\blacksquare$
 
\end{appendices}



\end{document}